\documentclass[journal]{IEEEtran}
\usepackage{cite}
\usepackage{amsmath}
\usepackage{graphicx}
\usepackage{float}
\usepackage{caption}
\usepackage{subfigure}
\usepackage{color}
\usepackage{CJK}
\usepackage{algorithm}
\usepackage{algorithmic}
\usepackage{tikz}
\usepackage{epstopdf}
\usepackage{indentfirst}
\usepackage{subfigure}
\usepackage{epstopdf}
\usepackage[latin1]{inputenc}
\usepackage{rotating}
\usepackage{colortbl}
\usepackage{amssymb,amsthm,enumerate,graphicx,float}
\usepackage{booktabs}
\usepackage{hyperref}
\pdfoutput=1
\ifCLASSINFOpdf
\else
\fi
\newtheorem{theorem}{Theorem}

\newtheorem{definition}{Definition}
\newtheorem{lemma}{Lemma}

\newtheorem{assumption}{Assumption}

\newtheorem{remark}{Remark}

\begin{document}
\title{Distributed Generalized Nash Equilibria Seeking Algorithms Involving Synchronous and Asynchronous Schemes}
\author{Huaqing Li, \IEEEmembership{Senior Member, IEEE}, Liang Ran, Lifeng Zheng, Zhe Li, Jinhui Hu, Jun Li, Tingwen Huang, \IEEEmembership{Fellow, IEEE}
\thanks{The work described in this paper was supported in part by the National Natural Science Foundation of China under Grant 62173278. \emph{(Corresponding author: Huaqing Li.)}

Huaqing Li, Liang Ran, Lifeng Zheng, Zhe Li and Jun Li are with the Chongqing Key Laboratory of Nonlinear Circuits and Intelligent Information Processing, College of Electronic and Information Engineering, Southwest University, Chongqing 400715, China (e-mail: huaqingli@swu.edu.cn, ranliang\_rl@163.com, zlf\_swu@163.com, lizhe\_0415@163.com, jun\_li2023@163.com).

J. Hu is with the Department of Automation, Central South University, Changsha 410083, P. R. China. (e-mail: jinhuihu@csu.edu.cn). He is also with the Department of Biomedical Engineering, City University of Hong Kong, Hong Kong SAR, P. R. China (e-mail:jinhuihu3-c@my.cityu.edu.hk).

Tingwen Huang is with Texas A \& M University at Qatar, Doha 23874, Qatar. (e-mail: tingwen.huang@qatar.tamu.edu).

}}

\maketitle

\begin{abstract}
This paper considers a class of noncooperative games in which the feasible decision sets of all players are coupled together by a coupled inequality constraint. Adopting the variational inequality formulation of the game, we first introduce a new local edge-based equilibrium condition and develop a distributed primal-dual proximal algorithm with full information. Considering challenges when communication delays occur, we devise an asynchronous distributed algorithm to seek a generalized Nash equilibrium. This asynchronous scheme arbitrarily activates one player to start new computations independently at different iteration instants, which means that the picked player can use the involved out-dated information from itself and its neighbors to perform new updates. A distinctive attribute is that the proposed algorithms enable the derivation of new distributed forward-backward-like extensions. In theoretical aspect, we provide explicit conditions on algorithm parameters, for instance, the step-sizes to establish a sublinear convergence rate for the proposed synchronous algorithm. Moreover, the asynchronous algorithm guarantees almost sure convergence in expectation under the same step-size conditions and some standard assumptions. An interesting observation is that our analysis approach improves the convergence rate of prior synchronous distributed forward-backward-based algorithms. Finally, the viability and performance of the proposed algorithms are demonstrated by numerical studies on the networked Cournot competition.
\end{abstract}

\IEEEpeerreviewmaketitle

\section{Introduction}

\IEEEPARstart{A}{long} with the penetration of multiagent networks, noncooperative games have  played a more and more active role in many multiagent applications, such as operations research community \cite{Ruud2020Solving}, price-based demand response \cite{Mao2017Price}, energy management \cite{Liang2019Generalized} and automatic control \cite{Ankur2012variational}. In noncooperative games, each player (also called agent) individually controls its own decision, and is also equipped with a cost function determined jointly by decisions of all players. The primary goal of each player is choosing a best decision to optimize its cost selfishly. Generally speaking, the best-response decisions of all players are referred as to the generalized Nash equilibria (GNE) for the game \cite{Patrick1991Generalized,Francisco2007generalized}.

\subsection{Literature Review}

In the past decades, great efforts have been made to seek the GNE. In work \cite{Kaemmler2015GENERALIZED}, an approximation scheme with a path-following strategy was presented to address standard Nash equilibria (NE) problems. In work \cite{Jason2009Joint} to consider potential games, learning-based algorithms were proposed. Due to the limitations of centralized algorithms \cite{Kaemmler2015GENERALIZED,Jason2009Joint}, such as limited flexibility and prohibitively expensive computational and communication, insightful distributed computing frameworks for seeking GNE have increasingly gained attention \cite{Mattia2020Proximal,Belgioioso2021time}.

For games with feasible constraints, a NE seeking strategy that combines projection operations with consensus dynamics was designed in \cite{Distributed2021Yanan}. In work \cite{Maojiao2022Differentially}, a new distributed gradient descent method was presented to minimize the private objectives, but a decaying step-size was used. Meanwhile, a gradient-based updating law was further introduced in \cite{Yipeng2021Distributed} and \cite{Fei2021Distributed}, wherein the former considered constant step-sizes. Inspired by proximal operators, an inexact alternating direction method of multipliers (ADMM) was proposed in \cite{Salehisadaghiani2019Seeking} to pursue NE under partial-decision information. To solve noncooperative games subject to coupled equality constraints, a distributed primal-dual algorithm using leader-following consensus protocol was put forward in \cite{Kaihong2021Nonsmooth}, and another gradient-based projected algorithm was reported in \cite{Kaihong2019Distributed}. Also, in work \cite{Liang2017Coupled}, non-smooth tracking dynamics were leveraged to design an iterative algorithm where a group of players can take simple local information exchange. Note that these algorithms are continuous-time designs. For seeking GNE in discrete-time domain, a distributed primal-dual algorithm as well as its relaxed version was derived in \cite{Yi2019splitting}. In work \cite{Franci2021Distributed}, an elegant distributed forward-backward algorithm was derived to solve the stochastic game. Note that both \cite{Yi2019splitting} and \cite{Franci2021Distributed} take advantage of the forward-backward operator splitting technique to establish the convergence. Other effective splitting techniques can be found in \cite{Pavel2020Distributed,Zhaojian2021Energy}.

All above methods seeking the GNE are confined to synchronous updates, which demands a global coordinated clock and necessitates the utilization of the latest data for next iterations. However, the idle time caused by the communication asynchrony among all players will be inevitable, which leads to a waste of time resource and decreases the efficiency. In contrast, asynchronous implementations have significant advantages, including reducing idle time of communication links, saving computing, and enhancing the robustness of algorithms \cite{Tian2020Achieving}. Recently, related asynchronous algorithms have been proposed for solving cooperative convex optimization problems, such as ASY-SONATA \cite{Tian2020Achieving} and asynchronous PG-EXTRA \cite{Tianyu2018Decentralized}, but not available for noncooperative games. To solve games asynchronously, few efforts have been made in \cite{Lei2017Randomized,Peng2020Asynchronous,Cenedese2019An}. Specifically, an asynchronous inexact proximal best-response scheme was developed in \cite{Lei2017Randomized}, but only for NE seeking. By employing forward-backward splitting technique, two asynchronous distributed GNE seeking algorithms were  presented in \cite{Peng2020Asynchronous}, and another node-based variant was designed in \cite{Cenedese2019An}. However, it is essential to note that: i) the asynchronous method in \cite{Peng2020Asynchronous} under edge-based modes necessitates neighboring players to commonly share edge information, not implementable for independent management of each player's dual information; ii) the asynchronous approaches in \cite{Cenedese2019An} still require a coordinator to collect delayed dual information, which is a challenge of supporting the distributed implementation; iii) although asynchronous convergence is studied in \cite{Peng2020Asynchronous,Cenedese2019An}, the technical challenge to analyze the asynchronous error caused by communication delays remains unclear. It is of great significance to cope with the above challenging issues, which makes contributions for a broad game theory of asynchronous distributed GNE seeking algorithms.

\subsection{Contribution Statements}
In general, the key contributions of this paper are summarized in the following three aspects.

(1) This paper concentrates on the generalized noncooperative games subject to local constraints as well as coupled inequality one. Following the Karush-Kuhn-Tucker (KKT) condition of variational inequalities \cite{Francisco2010Generalized}, a new local equilibrium condition utilizing edge-based dual profiles is first presented. In comparison with these Laplacian-based conditions \cite{Fei2021Distributed,Salehisadaghiani2019Seeking,Liang2017Coupled,Yi2019splitting,Franci2021Distributed,Belgioioso2017Decentralized,Pavel2020Distributed,Peng2020Asynchronous,Cenedese2019An}, this new condition eliminates the (weighted) Laplacian matrices or weight balancing strategies , supporting completely independent data management for each player.

(2) Given the provided new equilibrium condition, a distributed primal-dual proximal algorithm using proximal operators (DPDP\_GNE for short) is first developed. To tackle the delayed issue, the asynchronous iterative algorithm of DPDP\_GNE (ASY\_DPDP\_GNE for short) is designed, which wakes up one random player to perform updates with possibly involved out-dated data, while the inactive players remain unchanged but are allowed to exchange information with its neighbors. Note that ASY\_DPDP\_GNE requires less neighbor information than the asynchronous method in \cite{Peng2020Asynchronous}. More interestingly, our frameworks can readily recover new forward-backward-like algorithms, essentially different from recent forward-backward methods \cite{Yi2019splitting,Franci2021Distributed,Pavel2020Distributed,Peng2020Asynchronous}.

(3) In theoretical aspect, we provide a new analysis approach, showing that the sequence generated by DPDP\_GNE is quasi-Fej{\'{e}}r monotone and its fixed-point residual achieves $o(1/(k))$-convergence, provided that uncoordinated step-sizes are chosen from explicit ranges. Under the same step-size conditions as DPDP\_GNE, we introduce an additional new metric to absorb the delays and explicitly derive monotonic conditional expectations associated with the distances to the fixed point to illustrate the convergence of ASY\_DPDP\_GNE. Notably, our new analysis generalizes these results \cite{Yi2019splitting,Pavel2020Distributed,Peng2020Asynchronous,Cenedese2019An}, and gives improved convergence rates of their synchronous iterations (see Section \ref{DiscTheore}).

\subsection{Organization}
Section \ref{SecFormu} formulates the game model and describes its variational inequality setting. Section \ref{SecAlgoms} develops distributed GNE seeking algorithms, including synchronous and asynchronous versions. Then, the convergence analysis of both proposed algorithms is provided in Section \ref{SecConverg}, and numerical experiments are carried out to evaluate the performance in Section \ref{SecPerfo}. Finally, Section \ref{SecConcl} draws some concluding remarks.

\subsection{Notations}

Throughout this paper, let $\mathbb{N}$, ${\mathbb{R}^n}$ and ${\mathbb{R}^{m \times n}}$ denote the set of non-negative integers, the set of $n$-dimension real vectors and the set of $m \times n $ real matrices, respectively. The non-negative real set is $\mathbb{R}_ + ^n = \{ x \in {\mathbb{R}^n}|{x_1},{\text{ }} \cdots ,{\text{ }}{x_n} \ge 0\} $ with the elements ${x_1},{\text{ }} \cdots ,{\text{ }}{x_m}$ of the vector $x$. Use ${\rm{col}}\{ {x_1},{\text{ }} \cdots ,{\text{ }}{x_m}\} $ to represent the column vector stacked with ${x_1},{\rm{ }}{x_2},{\rm{ }} \cdots {,\rm{ }}{x_m}$, and ${\rm{blkdiag}}\{ {X_1},{\text{ }} \cdots ,{\text{ }}{X_m}\}$ to denote the block diagonal matrix composed of blocks ${X_1},{\rm{ }} \cdots {,\rm{ }}{X_m}$. Denote the condition number of a positive definite matrix $V \in {\mathbb{R}^{n \times n}}$ by $\kappa (V) = {\lambda _{\max }}(V)/{\lambda _{\min }}(V)$, where ${\lambda _{\max }}(V)$ and ${\lambda _{\min }}(V)$ are the maximum and minimum eigenvalues, respectively. For $x,y \in {\mathbb{R}^n}$, let ${\left\langle {x,y} \right\rangle _V} = \left\langle {x,Vy} \right\rangle $ and ${\left\| x \right\|_V} = \sqrt {{{\left\langle {x,x} \right\rangle }_V}}$ represent the inner product and the induced norm, respectively. Considering a non-empty closed convex set ${\mathbf{C}}$, its indicator function is defined by ${\delta _{\mathbf{C}}}(x) =0$ if $x \in {\mathbf{C}}$, otherwise $+ \infty$. Its normal cone is given by ${N_{\mathbf{C}}}(x) = \{ v \in {\mathbb{R}^n}|\langle {v,y - x} \rangle  \leq 0,{\text{ }}\forall y \in {\mathbf{C}}\}$, and its projection operator by ${\mathcal{P}_{\mathbf{C}}}(x) = {\text{argmi}}{{\text{n}}_{y \in {\mathbf{C}}}}{\left\| {y - x} \right\|^2}$.
For a proper, closed, convex function $f:{\mathbb{R}^n} \to \mathbb{R}$, the set of its sub-differential is described as $\partial f(x) = \{ v \in {\mathbb{R}^n}|f(x) + \langle {v,y - x} \rangle  \leq f(y),{\text{ }}\forall y \in {\mathbb{R}^n}\}$, and its conjugate is defined by ${f^ \star }(v) = {\text{sup}}{_{x \in {\mathbb{R}^n}}}\{{v^{\rm T}}x - f(x)\}$. The proximal operator of $f$ is defined by ${\rm{prox}}_f^V(x) = {\rm{argmi}}{{\text{n}}_{z \in {\mathbb{R}^n}}}\{ f(z) + 1/2||x - z||_V^2\} $, where $V$ is positive definite. For an operator $\mathcal{A}:{\mathbb{R}^n} \to {\mathbb{R}^{\tilde n}}$, its zeros and fixed points are denoted by ${\rm{zer (}}\mathcal{A}{\text{)}} = \left\{ {x \in {\mathbb{R}^n}|0 \in \mathcal{A}x} \right\}$ and ${\rm{fix (}}\mathcal{A}{\text{)}} = \left\{ {x \in {\mathbb{R}^n}|x \in \mathcal{A}x} \right\}$, respectively.

\section{Game Formulation and Development} \label{SecFormu}
This section considers a generalized noncooperative game with $m$ players in a communication graph $\mathcal{G}( {\mathcal{V},{\rm{ }}\mathcal{E}} )$, where $\mathcal{V} = \{ {1, \cdots ,m} \}$ is the set of players, and $\mathcal{E} \subseteq \mathcal{V} \times \mathcal{V}$ is the set of edges. Each player $i \in \mathcal{V}$ privately determines its own decision variable ${x_i} \in {\Omega _i}$, where ${\Omega _i} \subseteq {\mathbb{R}^{{n_i}}}$ is a feasible set. Let $x = {\text{col}}\{ {x_1}, \cdots ,{x_m}\}  \in {\mathbb{R}^n}$ be the aggregated variable, with the dimension $n = \sum\nolimits_{i = 1}^m {{n_i}}$. To highlight the $i$-th player's decision within $x$, we rewrite $x$ as $x = ({x_i},{x_{ - i}})$, where ${x_{ - i}} \in {\mathbb{R}^{n - {n_i}}}$ is a stacked vector except the component ${x_i}$. In the game, each player $i \in \mathcal{V}$ privately manages an objective function ${f_i}:{\mathbb{R}^n} \to \mathbb{R}$, in which its own decision is subject to ${\mathcal{X}_i}\left( {{x_{ - i}}} \right) = \left\{ {\left. {{x_i} \in {\Omega _i}} \right|\left( {{x_i},{x_{ - i}}} \right) \in \mathcal{X}} \right\}$. The coupled constraint is given by $\mathcal{X} = \prod\nolimits_{i = 1}^m {{\Omega _i}}  \cap \left\{ {\left. {x \in {\mathbb{R}^n}} \right|\sum\nolimits_{i = 1}^m {{A_i}{x_i}}  \leq \sum\nolimits_{i = 1}^m {{b_i}} } \right\}$ with the local transformer ${A_i}:{\mathbb{R}^q} \to {\mathbb{R}^{{n_i}}}$ and constant vector ${b_i} \in {\mathbb{R}^q}$. For the given decision ${x_{ - i}}$, the goal of player $i \in \mathcal{V}$ is to choose its best-response decision $x_i$ in order to optimize
\begin{flalign}\label{GENProblem}
  \mathop {{\text{minimize}}}\limits_{{x_i} \in {\mathbb{R}^{{n_i}}}} {\text{  }}{f_i}\left( {{x_i},{x_{ - i}}} \right) ,{\text{ }} {\text{subject to  }}{x_i} \in {\mathcal{X}_i}\left( {{x_{ - i}}} \right)
\end{flalign}

For the sequel analysis, we abbreviate the game \eqref{GENProblem} as GNG($\Xi $), and express its solution set (i.e., so-called generalized Nash equilibria) as GNE($\Xi$).
\begin{assumption}\label{AssumSetconvex}
(i) For player $i \in \mathcal{V}$, ${\Omega _i}$ is nonempty, compact and convex; (ii) $A_i$ is linear and bounded, and there exists ${{\hat x}_1} \in {\Omega _1}, \cdots ,{{\hat x}_m} \in {\Omega _m}$ such that $\sum\nolimits_{i = 1}^m {{A_i}{\hat x_i}}  \le \sum\nolimits_{i = 1}^m {{b_i}} $; (iii) Given all feasible $x_{-i}$, ${\mathcal{X}_i}({x_{ - i}})$ is nonempty, compact and convex; (iv) For all feasible $x_{-i}$, ${f_i}({x_i},{x_{ - i}})$ is convex and continuously differentiable at $x_i$.
\end{assumption}

Let ${x^ * } = {\text{col}}\{ x_1^ * , \cdots ,x_m^ * \} $ be a GNE of GNG($\Xi $) if it holds that ${f_i}(x_i^ * ,x_{ - i}^ * ) \leq {f_i}({x_i},x_{ - i}^ * )$ for any ${x_i} \in {\mathcal{X}_i}(x_{ - i}^ * )$. In other words, $x^* \in {\rm{GNE}}(\Xi )$ if and only if there exists a multiplier $\tilde u_i^ *  $ for player $i \in \mathcal{V}$ such that \cite[Theorem 8]{Francisco2010Generalized}
\begin{subequations}\label{GNEKKT}
  \begin{equation}\label{GNEKKTa}
   0 \in {\nabla _{{x_i}}}{f_i}\left( {x_i^ * ,x_{ - i}^ * } \right) + {N_{{\Omega _i}}}(x_i^ * ) + A_i^{\rm T}\tilde u_i^ *
  \end{equation}
  \begin{equation}\label{GNEKKTb}
    0 \in {N_{\mathbb{R}_ + ^q}}(\tilde u_i^ *  ) - \sum\limits_{i = 1}^m {\left( {{A_i}x_i^ *  - {b_i}} \right)} \qquad \quad\,\,
  \end{equation}
\end{subequations}

\begin{remark} \label{LabelAssump1}
Assumption \ref{AssumSetconvex} (i) and (ii) ensure that $\mathcal{X}$ is nonempty, compact and convex, and it is generally interpreted as the bounded property. This is a basic premise for GNG($\Xi $) to get its GNE. Overall, Assumption \ref{AssumSetconvex} guarantees the existence of GNE for game \eqref{GENProblem}. It is a common and standard assumption in \cite{Francisco2007generalized,Liang2019Generalized,Peng2020Asynchronous,Kaihong2019Distributed}.
\end{remark}

\subsection{Practical Scenarios} \label{SubPrSce}
Many practical scenarios can be modeled as the GNG($\Xi $). This section provides two typical practical applications.

\emph{1) Cournot Market Competition \cite{Zheng2023Stochastic,Kannan2012Iterative}}: Consider a factory set $\mathcal{W}$, a commodity set $\mathcal{D}$ and a purchaser set $\mathcal{R}$. Each factory $i \in \mathcal{W}$ produces its commodity $j\in {\mathcal{D}_i}$ and chooses $q_{j,s}^i \in {\Omega _{i,j}}$ amount of commodity to sell to different purchasers $s\in {\mathcal{R}_i}$, where ${\Omega _{i,j}} = \{ 0 \leq q_{j,s}^i \le q_{j,\max }^i\} $ is the production limitation with upper bounder $q_{j,\max }^i$, ${\mathcal{R}_i} \subset \mathcal{R}$ is the commodities set that factory $i$ can produce, and ${\mathcal{D}_i} \subset \mathcal{D}$ is the purchaser set who has a purchasing relation with factory $i$. Once factory $i$ sells its goods, it can gain a revenue ${f_{i,1}}\left( {{q_i},{q_{ - i}}} \right) = \sum\nolimits_{j \in {\mathcal{D}_i}} {\sum\nolimits_{s \in {\mathcal{R}_i}} {\left( {{g_{j,s}} - {\rho _{j,s}}\sum\nolimits_{v \in {\mathcal{Y}_s}} {q_{j,s}^v} } \right)q_{j,s}^v} } $, but it will take the production cost ${f_{i,2}}\left( {{q_i}} \right) = \sum\nolimits_{j \in {\mathcal{D}_i}} {\sum\nolimits_{s \in {\mathcal{R}_i}} {\left( {{a_{i,s}}q_{j,s}^i \cdot q_{j,s}^i + {c_{i,s}}q_{j,s}^i} \right)} }$, where ${q_i} = {\rm {col}}{\left\{ {q_{j,s}^i} \right\}_{j \in {\mathcal{D}_i},s \in {\mathcal{R}_i}}}$ is the production decision determined by factory $i$, ${{g_{j,s}}}$ is purchaser $s$'s preset purchase price for commodity $j$, ${{\rho _{j,s}}}$ is the pricing parameter, ${{\mathcal{Y}_{j,s}}}$ is the set of factories selling commodity $j$ to purchaser $s$, ${{a_{i,s}}}$ and ${{c_{i,s}}}$ are the cost parameters. Notice that all purchasers have the limited storage capacity $\tilde b = {\rm {col}}{\{ {{{\tilde b}_s}} \}_{s \in \mathcal{R}}}$ for all total commodity, leading to the coupled constraint $\sum\nolimits_{i = 1}^m {{A_i}{q_i}}  \leq \tilde b$, where $A_i$ represents the supply relationship defined by the procurement relationship. In effect, factory $i$' payoff function is parametrized by the total sales ${\sum\nolimits_{v \in {\mathcal{Y}_{j,s}}} {q_{j,s}^v} }$, thereby leading to the noncooperative game. Therefore, by setting ${f_i}\left( {{q_i},{q_{ - i}}} \right) = {f_{i,2}}\left( {{q_i}} \right) - {f_{i,1}}\left( {{q_i},{q_{ - i}}} \right)$, the game of the market competition is modeled as
\begin{flalign*}
\mathop {\text{minimize}}\limits_{{q_i} \in \prod\nolimits_{j \in {\mathcal{D}_i}} {{\Omega _{i,j}}}} {\text{ }}{f_{i}}\left( {{q_i},{q_{ - i}}} \right),{\text{ }} {\text{subject to}}  \sum\nolimits_{i = 1}^m {{A_i}{q_i}}  \le \tilde b
\end{flalign*}

\emph{2) Demand Response Management \cite{Mao2017Price}}: There are $m$ price anticipating users in an electricity market. Each user is equipped with an energy-management controller to schedule the electricity usage, and with an advanced metering infrastructure to bidirectionally communicate among electricity users and an energy provider. In practical, each user $i$ needs to schedule its energy consumption $p_i$ within its acceptable range ${\mathcal{X}_i} = \left\{ {{p_i}\left| {{p_{i,\min }} \leq {p_i} \leq {p_{i,\max }}} \right.} \right\}$ to meet the total load demand, i.e., $\sum\nolimits_{i = 1}^m {{p_i}}  = \sum\nolimits_{i = 1}^m {{d_i}} $. Under this setting, user $i$ minimizes its cost ${C_i}\left( {{p_i},{p_{ - i}}} \right) = {C_{i,1}}\left( {{p_i}} \right) + {C_{i,2}}\left( {{p_i},{p_{ - i}}} \right){p_i}$, where ${C_{i,1}}\left( {{p_i}} \right)$ is the load curtailment cost, and ${C_{i,2}}\left( {{p_i},{p_{ - i}}} \right)$ is the price function determined by total energy consumption. Note that the price function leads to the noncooperative game. Hence, by setting ${A_i} = {[ {1, - 1} ]^{\rm T}}$ and ${b_i} = {[ {{{ d}_i}, - {{ d}_i}} ]^{\rm T}}$, the game of the demand response management is modeled as
\begin{flalign*}
\mathop {\text{minimize}} \limits_{{p_i} \in {\mathcal{X}_i}} {\text{ }}{C_i}\left( {{p_i},{p_{ - i}}} \right),{\text{ }} {\text{subject to}} \sum\nolimits_{i = 1}^m {{A_i}{p_i}}  \le \sum\nolimits_{i = 1}^m {{b_i}}
\end{flalign*}

\begin{remark}
Note that recent works \cite{Kaihong2019Distributed,Peng2020Asynchronous} studied the noncooperative games by developing effective distributed GNE seeking algorithms, but only coupled equality constraints are involved. In contrast, the considered coupled inequality constraint in game \eqref{GENProblem} can address the equality one by means of double-sided inequalities (see the second example), and therefore is more general than \cite{Kaihong2019Distributed,Peng2020Asynchronous}.
\end{remark}

\subsection{Variational inequality formulation}

It remains a challenge to directly compute the GNE of GNG($\Xi $) via the condition \eqref{GNEKKT}. Under Assumption \ref{AssumSetconvex}, an (sufficient) equilibrium condition of GNG($\Xi $) is specified as a variational inequality (VI) formulation \cite{Kaihong2021Nonsmooth,Kaihong2019Distributed,Liang2017Coupled,Mattia2020Proximal}. We first define the pseudogradient associated with ${\nabla _{{x_i}}}{f_i}({x_i},{x_{ - i}})$ as $F(x) = {\text{col}}\{ {\nabla _{{x_1}}}{f_1}({x_1},{x_{ - 1}}), \cdots ,{\nabla _{{x_m}}}{f_m}({x_m},{x_{ - m}})\}$. Then the VI problem is to find ${x^ * } \in \mathcal{X}$, such that
\begin{flalign}\label{VIproblem}
\left\langle {F\left( {{x^ * }} \right),x - {x^ * }} \right\rangle  \ge 0,{\text{ }}\forall x \in \mathcal{X}
\end{flalign}

Recalling that $\mathcal{X}$ is nonempty, compact and convex as discussed in Remark \ref{LabelAssump1}. Since ${f_i}({x_i},{x_{ - i}})$ is convex in $x_i$ for any feasible $x_{-i}$, the pseudogradient $F$ is monotone \cite{Scutari2010Convex} and the solution set of problem \eqref{VIproblem} is nonempty, compact and convex \cite[Lemma 3]{Liang2019Generalized}.

\begin{assumption}\label{AssumpseudoLipfucnti}
The pseudogradient $F(x)$ is $L_F$-Lipschitz continuous over $\mathcal{X}$, i.e., $\left\| {F(x) - F(y)} \right\| \leq {L_F}\left\| {x - y} \right\|$ for any $x,y \in \mathcal{X}$. And $F(x)$ is strong monotonicity with $\mu $ over $\mathcal{X}$, i.e., $\langle {x - y,F(x) - F(y)} \rangle  \ge \mu {\left\| {x - y} \right\|^2}$ for any $x,y \in \mathcal{X}$.
\end{assumption}

In brief, we denote the solution set of VI problem \eqref{VIproblem} as ${\rm{SOL}}\left( {\mathcal{X},F} \right)$. Hence, it follows from Lagrangian duality \cite{Facchinei2009Nash} that $x^* \in {\rm{SOL}}\left( {\mathcal{X},F} \right) $ if and only if there exists a global multiplier $u_g^* \in {\mathbb{R}^q}$ such that
\begin{subequations}\label{VIKKT}
  \begin{equation}\label{VIKKTa}
   0 \in {\nabla _{{x_i}}}{f_i}\left( {x_i^ * ,x_{ - i}^ * } \right) + {N_{{\Omega _i}}}(x_i^ * ) + A_i^{\rm T}u _g^ *
  \end{equation}
  \begin{equation}\label{VIKKTb}
    0 \in {N_{\mathbb{R}_ + ^q}}(u_g^ * ) - \sum\limits_{i = 1}^m {\left( {{A_i}x_i^ *  - {b_i}} \right)} \qquad \quad\,\,
  \end{equation}
\end{subequations}

By comparing conditions in \eqref{GNEKKT} and \eqref{VIKKT}, one can confirm that a solution of problem \eqref{VIproblem} is also a GNE of GNG($\Xi$) provided that all multipliers in \eqref{GNEKKT} commonly share the same value, i.e., $\tilde u _1^ *  =  \cdots  = \tilde u _m^ *  = u_g^*$. The next lemma ties with the relationship between GNE($\Xi$) and ${\rm{SOL}}\left( {\mathcal{X},F} \right)$.

\begin{lemma} \cite[Theorem 3.1]{Francisco2007generalized} \label{SoluVItoGNE}
Let Assumption \ref{AssumSetconvex} hold. Every solution ${x^ * }$ of VI is also a GNE of the game \eqref{GENProblem}, i.e., ${x^ * } \in {\rm{SOL}}\left( {\mathcal{X},F} \right) \Rightarrow {x^ * } \in {\rm{GNE}}\left({{\Xi}} \right)$.
\end{lemma}

\begin{remark}
Based on Assumption \ref{AssumSetconvex}, the monotonicity of $F$ as well as the convexity, nonempty and compactness of ${\cal X}$ contribute to the existence of ${\rm{SOL}}\left( {\mathcal{X},F} \right)$ \cite{Francisco2010Generalized}. The strong monotonicity of $F$ in Assumption \ref{AssumpseudoLipfucnti} further allows us to claim the uniqueness of ${\rm{SOL}}\left( {\mathcal{X},F} \right)$ \cite{Scutari2010Convex}.
\end{remark}

\subsection{Local Equilibrium Condition}

This section extends the condition \eqref{VIKKT} and describes a new local equilibrium condition by introducing a group of local dual profiles ${u_1}, \cdots ,{u_m}$. Players maintain local multipliers and exchange them with neighbors over $\mathcal{G}$, in order to learn the Lagrange multiplier $u_g^*$. We make the following assumption on communication network.
\begin{assumption}\label{AsCommu}
The graph $\mathcal{G}$ is undirected and connected.
\end{assumption}

Based on Assumption \ref{AsCommu}, to highlight the edge $\left( {i,j} \right) \in {\cal E}$ linking with players $i$ and $j$, we define the matrix ${E_{ij}} = I \in {\mathbb{R}^{q \times q}}$ if $i<j$, and ${E_{ij}} = -I \in {\mathbb{R}^{q \times q}}$ if $i>j$. Then, neighboring players $i$ and $j$ learn the Lagrange multiplier $u_g^*$ over the following edge-based consensus constraint:
\begin{flalign}\label{EijEji}
{E_{ij}}{u_i} + {E_{ji}}{u_j} = 0,{\text{ }}\forall (i,j) \in \mathcal{E}
\end{flalign}
To proceed, we define the set ${\mathcal{C}_{\left( {i,j} \right)}} = \left\{ {\left( {{z_1},{z_2}} \right) \in {\mathbb{R}^q} \times {\mathbb{R}^q}\left| {{z_1} + {z_2} = 0} \right.} \right\}$ and the mapping ${\Pi_{\left( {i,j} \right)}}:{\text{ }}u \to \left( {{E_{ij}}{u_i},{\text{ }}{E_{ji}}{u_j}} \right)$, where $u = {\text{col}}\{ {u_1}, \cdots ,{u_m}\}  \in {\mathbb{R}^{mq}}$. Then, the coupled constraint \eqref{EijEji} can be equivalently cast as the set form ${\Pi _{\left( {i,j} \right)}}u \in {\mathcal{C}_{\left( {i,j} \right)}}$.

For each edge $\left( {i,j} \right) \in \mathcal{E}$, we consider the variable ${w_{(i,j)}} = {\text{col}}\{ {w_{(i,j),i}},{w_{(i,j),j}}\}  \in {\mathbb{R}^{2q}}$, where ${w_{(i,j),i}} \in {\mathbb{R}^q}$ and ${w_{(i,j),j}} \in {\mathbb{R}^q}$ are respectively kept by $i$ and $j$, used for reaching consensus of the local multipliers. Consequently, a new local equilibrium condition can be established via the following lemma.
\begin{lemma}\label{DistriMN}
Let Assumptions \ref{AssumSetconvex}-\ref{AsCommu} hold. Consider the point $\left( {{x^ * },{u^ * },{w^ * }} \right)$, where ${x^ * } = {\rm{col}}\{ x_i^ * \} _{i = 1}^m$, ${u^ * } = {\rm{col}}\{ u_i^ * \} _{i = 1}^m$ and ${w^ * } = {\rm{col}}{\{ w_{\left( {i,j} \right)}^ * \} _{\left( {i,j} \right) \in \mathcal{E}}}$ with $w_{\left( {i,j} \right)}^ *  = {\rm{col}}\{ w_{\left( {i,j} \right),i}^ * ,w_{\left( {i,j} \right),j}^ * \} $. If each player's local components $(x_i^*,u_i^*,{\{ w_{(i,j),i}^*\} _{j \in {{\cal N}_i}}})$ meet
\begin{subequations}\label{DisMNcondit}
  \begin{equation}\label{DisMNcondita}
   0 \in {\nabla _{{x_i}}}{f_i}\left( {x_i^ * ,x_{ - i}^ * } \right) + {N_{{\Omega _i}}}(x_i^ * ) + A_i^{\rm T}u_i^ * \qquad\quad
  \end{equation}
  \begin{equation}\label{DisMNconditb}
  0 \in {N_{\mathbb{R}_ + ^q}}(u_i^ * ) + {b_i} - {A_i}x_i^ *  + \sum\nolimits_{j \in {\mathcal{N}_i}} {E_{ij}w_{\left( {i,j} \right),i}^ * }
  \end{equation}
  \begin{equation}\label{DisMNconditc}
   0 \in \partial \delta _{{\mathcal{C}_{\left( {i,j} \right)}}}^ \star ( {w_{( {i,j} )}^ * } ) - {\Pi_{\left( {i,j} \right)}}{u^ * } \qquad\quad\qquad\quad\quad\,\,\,
  \end{equation}
\end{subequations}
Then, $u_1^ * , \cdots ,u_m^ * $ reach consensus, and ${x^ * } \in {\rm{SOL}}\left( {\mathcal{X},F} \right)$.
\end{lemma}
\begin{proof}
See Appendix \ref{ProofDistriMN}.
\end{proof}

\begin{remark}
Note that ${\Pi_{\left( {i,j} \right)}}{u^ * } = {\rm{col}} ( {{E_{ij}}u_i^ * ,{E_{ji}}u_j^ * } )$, the new condition \eqref{DisMNcondit} is thus totally local. In spite of well-known Laplacian-based condition in \cite{Cenedese2019An,Peng2020Asynchronous}, the compact expression is studied and local condition for each player is difficult to be obtained due to the coupled structure of Laplacian matrices. By contrast, the edge-based form \eqref{EijEji} facilitates the explicit presentation of condition \eqref{DisMNcondit}, thereby contributing to the development of distributed algorithms. Lemma \ref{DistriMN} indicates that $(x^ * ,u^ * )$ satisfying \eqref{DisMNcondit} is always accessible to the condition system \eqref{VIKKT}. The relations among the previous formulations are summarized as: $(x^ * ,u^ * )$ satisfies \eqref{DisMNcondit} $\Rightarrow $ ${x^ * } \in {\rm{SOL}}\left( {\mathcal{X},F} \right)$ via Lemma \ref{DistriMN} $\Rightarrow $ ${x^ * } \in {\rm{GNE}}\left({{\Xi}} \right)$ via Lemma \ref{SoluVItoGNE}.
\end{remark}

\section{Distributed Algorithms For Seeking GNE} \label{SecAlgoms}

\subsection{Synchronous Algorithm} \label{SubSynchr}

In the following, the details of designing distributed algorithm for GNE are briefly mentioned below. We assume that each player can access decisions that its local objective function directly depends on, obtained via an additional communication graph (termed as interference graph) \cite{Yi2019splitting,Cenedese2019An}. Consequently, we specify an interference graph allowing each player to observe the decisions influencing its objective function and get its local gradient. Each player $i \in \mathcal{V}$ controls and iteratively updates the private variables $x_i$, $u_i$ and edge-based ones ${w_{(i,j),i}}$. To reach the consensus, player $i$ needs to conduct local information exchange with its neighbors, i.e., sharing the latest data ${E_{ij}}{u_i}$ and ${w_{(i,j),i}}$ with $j \in {\mathcal{N}_i}$. Here, we introduce three auxiliary variables, ${\bar w_{\left( {i,j} \right)}} = {\text{col}}\{ {\bar w_{\left( {i,j} \right),i}},{\bar w_{\left( {i,j} \right),j}}\} $, ${\bar u_i}$ and ${\bar x_i}$, corresponding to ${w_{\left( {i,j} \right)}}$, ${u_i}$ and ${x_i}$ respectively. By exploiting proximal operators, three auxiliary variables take the following recursions:
\begin{subequations}\label{Ppredphras}
  \begin{equation}\label{Ppredphrasa}
   {{\bar w}_{\left( {i,j} \right)}} = {\text{pro}}{{\text{x}}_{{\kappa _{\left( {i,j} \right)}}\delta _{{\mathcal{C}_{\left( {i,j} \right)}}}^ \star }}( {{w_{\left( {i,j} \right)}} + {\kappa _{\left( {i,j} \right)}}{\Pi_{\left( {i,j} \right)}}u} ) \qquad\quad\,\,\,
  \end{equation}
  \begin{equation}\label{Ppredphrasb}
  {{\bar u}_i} \!=\! {\text{pro}}{{\text{x}}_{\sigma_i \delta _{{\mathbb{R}_ + ^q}} }}( {{u_i} \!+\! {\sigma _i}( {{A_i}{x_i} \!-\! {b_i} \!-\!\! \sum\nolimits_{j \in {\mathcal{N}_i}} {E_{ij}{{\bar w}_{\left( {i,j} \right),i}}} } )} )
  \end{equation}
  \begin{equation}\label{Ppredphrasc}
   {{\bar x}_i} = {\text{pro}}{{\text{x}}_{{\tau _i}{\delta _{{\Omega _i}}}}}( {{x_i} - {\tau _i}( {{\nabla _{{x_i}}}{f_i}( {{x_i},{x_{ - i}}} ) + A_i^{\rm T}{{\bar u}_i}} )} ) \,\,\qquad
  \end{equation}
\end{subequations}
where ${\omega _{\left( {i,j} \right)}},{\sigma _i}$ and ${\tau _i}$ are positive constant step-sizes. Due to the accessibility of proximal operators of ${\delta _{\mathbb{R}_ + ^q}}$ and ${{\delta _{{\Omega _i}}}}$, \eqref{Ppredphrasb} and \eqref{Ppredphrasc} can reduce to
\begin{flalign}\label{Ppredphrasbarux}
\begin{gathered}
  {{\bar u}_i} = {\mathcal{P}_ {\mathbb{R}_ + ^q} }( {{u_i} \!+\! {\sigma _i}( {{A_i}{x_i} \!-\! {b_i} \!-\! \sum\nolimits_{j \in {\mathcal{N}_i}} {{E_{ij}}{{\bar w}_{\left( {i,j} \right),i}}} } )} ) \hfill \\
  {{\bar x}_i} = {\mathcal{P}_{{\Omega _i}}}\left( {{x_i} - {\tau _i}\left( {{\nabla _{{x_i}}}{f_i}\left( {{x_i},{x_{ - i}}} \right) + A_i^{\rm T}{{\bar u}_i}} \right)} \right) \hfill \\
\end{gathered}
\end{flalign}
In addition, with the help of Moreau decomposition \cite{Bauschke2011Convex}, the relationship \eqref{Ppredphrasa} for all $j \in {\mathcal{N}_i}$  can be decomposed into
\begin{flalign}\label{Ppredphrasbarw}
\begin{gathered}
  {{\bar w}_{\left( {i,j} \right)}} = {w_{\left( {i,j} \right)}} + {\kappa _{\left( {i,j} \right)}}{\Pi _{\left( {i,j} \right)}}u \hfill \\
  \qquad\quad\,\,\, - {\kappa _{\left( {i,j} \right)}}{\mathcal{P}_{{\mathcal{C}_{\left( {i,j} \right)}}}}(\kappa _{\left( {i,j} \right)}^{ - 1}{w_{\left( {i,j} \right)}} + {\Pi _{\left( {i,j} \right)}}u) \hfill \\
\end{gathered}
\end{flalign}
Note that the projection of $({z_1},{z_2}) \in {\mathbb{R}^n} \times {\mathbb{R}^n}$ onto ${{\mathcal{C}_{(i,j)}}}$ is derived as ${\mathcal{P}_{{\mathcal{C}_{(i,j)}}}}:({z_1},{z_2}) \to 1/2 ({z_1} - {z_2},{z_2} - {z_1})$. Recalling the mapping ${\Pi_{(i,j)}}$, and expanding \eqref{Ppredphrasbarw}, the iterative rule of ${\bar w_{\left( {i,j} \right),i}}$ takes
\begin{flalign}\label{Ppredphrasbarwiji}
{{\bar w}_{\left( {i,j} \right),i}} \!=\! \frac{1}
{2}({w_{\left( {i,j} \right),i}} \!+\! {w_{\left( {i,j} \right),j}}) \!+\! \frac{{{\kappa _{\left( {i,j} \right)}}}}
{2}({E_{ij}}{u_i} \!+\! {E_{ji}}{u_j})
\end{flalign}
Then, combining \eqref{Ppredphrasbarux} and \eqref{Ppredphrasbarwiji} yields the prediction phrase as presented in Algorithm \ref{alg:1}. On the other hand, we take a correction step and assign the obtained values to ${x_i}$, ${u_i}$ and $w_{\left( {i,j} \right),i}$, i.e., ${x_i} \leftarrow {{\bar x}_i}$, ${u_i} \leftarrow {{{\bar u}_i} + {\sigma _i}{A_i}\left( {{{\bar x}_i} - {x_i}} \right)}$ and ${w_{\left( {i,j} \right),i}} \leftarrow {{\bar w}_{\left( {i,j} \right),i}} + {\kappa _{(i,j)}}{E_{ij}}({{\bar u}_i}$$ - {u_i})$, which naturally are viewed as the iterative values at the clock $k+1$. As a result, the detailed pseudo-code of the distributed algorithm is exposed in Algorithm \ref{alg:1}.

\begin{algorithm}
	\renewcommand{\algorithmicrequire}{\textbf{Initialization:}}
	\renewcommand{\algorithmicensure}{\textbf{For} $k = 0,{\text{ }}1,{\text{ }} \cdots $
                                      \textbf{do:}\qquad\qquad\qquad\qquad\qquad\qquad\qquad}
	\caption{\textbf{:} DPDP\_GNE}
	\label{alg:1}
	\begin{algorithmic}[1]
		\REQUIRE Each player $i \in \mathcal{V}$ initializes $x_i^0 \in {\mathbb{R}^{{n_i}}}$, $u_i^0 \in{\mathbb{R}^{q}}$ and $ w_{\left( {i,j} \right),i}^0 \in {\mathbb{R}^{q}}$ for $j \in {\mathcal{N}_i}$. Meanwhile, it chooses positive step-sizes $\tau _i$, $\sigma_i$ and ${\omega _{(i,j)}}$ properly.
		\ENSURE
         Each player $i \in \mathcal{V}$ repeats, for $ j \in {\mathcal{N}_i}$,
		\STATE \textbf{Prediction phase:} \label{ste1} \\
               $\bar w_{\left( {i,j} \right),i}^k = \frac{1}{2}( {w_{\left( {i,j} \right),i}^k + w_{\left( {i,j} \right),j}^k} ) + \frac{{{\omega _{\left( {i,j} \right)}}}}{2}( {{E_{ij}}u_i^k + {E_{ji}}u_j^k} )$ \\

		       $\bar u_i^k = {\mathcal{P}_ {\mathbb{R}_ + ^q} }( {u_i^k + {\sigma _i}( {{A_i}x_i^k - {b_i} - \sum\nolimits_{j \in
                {\mathcal{N}_i}} {E_{ij}\bar w_{\left( {i,j} \right),i}^k} } )} )$ \\
        \STATE \textbf{Update phase:} \label{ste2} \\
                $x_i^{k + 1} = {\mathcal{P}_{{\Omega _i}}}( {x_i^k - {\tau _i}( {{\nabla _{{x_i}}}{f_i}( {x_i^k,x_{ - i}^k} ) + A_i^{\rm T}\bar u_i^k} )} )$ \\
                $u_i^{k + 1} = \bar u_i^k + {\sigma _i}{A_i}( {x_i^{k + 1} - x_i^k} )$ \\
		        $w_{( {i,j} ),i}^{k + 1} = \bar w_{( {i,j} ),i}^k + {\omega _{( {i,j} )}}{E_{ij}}( {\bar u_i^k - u_i^k})$
		\STATE \textbf{Transmission phase:} Player $i$ broadcasts ${{E_{ij}}u_i^{k + 1}}$ and $w_{\left( {i,j} \right),i}^{k + 1}$ to each neighbor $j \in {\mathcal{N}_i}$.
	\end{algorithmic}
\textbf{End.}
\end{algorithm}

Due to the edge-based variables, DPDP\_GNE can be carried out sequentially. At first, $\bar w_{\left( {i,j} \right),i}^k$ is computed to integrate the disagreement between $u_i^k$ and $u_j^k$, and in return it serves as the feedback signal and is further integrated into the update of $\bar u_i^k$ to reach consensus of local multipliers. After that, each player reads $ x_{-i}^k$ through interference graph and use $\bar u_i^k$ to accomplish the computation of $x_i^{k+1}$. Finally, $u_i^{k+1}$ and $ w_{\left( {i,j} \right),i}^{k+1}$ are calculated by using $x_i^{k+1}$ and $\bar u_i^k$. Therefore, DPDP\_GNE is called distributed due to the following features: i) it achieves full data decomposition because each player manages its own data; ii) compared with \cite[Algorithms 2 and 3]{Cenedese2019An}, it has no coordinator to collect edge-/node-based data to update.

\begin{remark} \label{ReAdvanEdge}
Note that our edge-based mode offers distinct advantages compared to related works: i) it allows neighboring players to independently keep edge-based data (i.e., $w_{(i,j),i}$ and $w_{(i,j),j}$), thus providing better support for privacy protection and local data management, unlike the edge-based rule \cite{Peng2020Asynchronous,Cenedese2019An} that still couples neighboring players and mandates them to access data commonly from shared edges; ii) it reduces the dependency on neighboring players' information, i.e., fewer neighbor details than \cite{Peng2020Asynchronous,Carnevale2022Carnevale}; iii) it frees DPDP\_GNE from the restrictive requirement of  (weighted) Laplacian matrices or weight balancing strategies, compared with \cite{Salehisadaghiani2019Seeking,Kaihong2019Distributed,Pavel2020Distributed,Franci2021Distributed,Peng2020Asynchronous,Cenedese2019An}. These advantages also happen in the following asynchronous algorithm.
\end{remark}

\subsection{Asynchronous Algorithm}

In the asynchronous setting, not all players participate in iterative updates. During the asynchronous implementation, we assume that each player has its own activation clock, which ticks according to an independent and identically distributed Poisson process. Having Poisson clocks is common in \cite{Mansoori2020Fast,Zhimin2016ARock}, which implies that the activation process of each player follows a Poisson process and is thus independently and identically distributed. In this case, we assume that only one clock ticks at the $k$-th iteration, and that the activated player is called as $i_k$. Therefore, $i_k$ holds the private variables ${x_{{i_k}}}$, ${u_{{i_k}}}$ and edge variables ${w_{({i_k},j),{i_k}}}$ for $j \in {\mathcal{N}_{{i_k}}}$, as well as the local parameters ${\sigma _{{i_k}}}$, ${\tau _{{i_k}}}$, ${\eta _{{i_k}}}$ and ${{\omega _{\left( {{i_k},j} \right)}}}$. Considering the lost or cluttered information caused by unreliable or unpredictable communication among players, the active player may possess a family of outdated information, i.e., ${x_{{i_k}}^{k - \phi _i^k}}$, ${u_{{i_k}}^{k - \phi _{{i_k}}^k}}$ and ${w_{\left( {{i_k},j} \right),{i_k}}^{k - \varphi _{\left( {{i_k},j} \right)}^k}}$, where ${\phi _{{i_k}}^k}$ and ${\varphi _{\left( {{i_k},j} \right)}^k}$ represent the delayed index. Let ${x_{ - {i_k}}^{k - \phi _{ - {i_k}}^k}}$ denote the delayed information consisting of all $x_j^{k - \phi _j^k}$, which influences the cost function ${f_{{i_k}}}( {{x_{{i_k}}},{x_{ - {i_k}}}}) $. To facilitate the asynchronous iterations, let each player maintain a local buffer, which stores involved outdated data. The delayed index needs to satisfy the following assumption.

\begin{assumption}\label{AssumpDelay}
At any iteration $k$, there exists a unified upper bound $\varepsilon >0$ such that $\phi _i^k,\varphi _{\left( {i,j} \right)}^k \leq \varepsilon $.
\end{assumption}

In view of the synchronous fashion in Algorithm \ref{alg:1}, the developed asynchronous algorithm consists of five phases: activation phase, receiving phase, prediction phase, update phase and transmission phase. More specifically, a player $i_k \in \mathcal{V}$ is picked arbitrarily at the instant $k$, and keeps receiving (possibly outdated) messages of ${w_{( {{i_k},j} ),j}}$, $E_{j{i_k}}{u_j}$ and ${x_{ - {i_k}}}$ from its neighbors via communication graph and interference graph. Then, it stores them into its local buffer and reads $x_{{i_k}}^{k - \phi _{{i_k}}^k}$, $w_{\left( {{i_k},j} \right),{i_k}}^{k - \varphi _{\left( {{i_k},j} \right)}^k}$ and all $E_{i_kj}{u_j^{k - \phi _j^k}}$ to compute
\begin{flalign}\label{Asywbariter}
\begin{gathered}
  \bar w_{\left( {{i_k},j} \right),{i_k}}^k = \frac{1}
{2}(w_{\left( {{i_k},j} \right),{i_k}}^{k - \varphi _{\left( {{i_k},j} \right)}^k} + w_{\left( {{i_k},j} \right),j}^{k - \varphi _{\left( {{i_k},j} \right)}^k}) \hfill \\
  \qquad\qquad\quad + \frac{{{\kappa _{\left( {{i_k},j} \right)}}}}
{2}({E_{{i_k}j}}u_{{i_k}}^{k - \phi _{{i_k}}^k} + {E_{j{i_k}}}u_j^{k - \phi _j^k}) \hfill \\
\end{gathered}
\end{flalign}
Subsequently, the active player further combines ${u_{{i_k}}^{k - \phi _{{i_k}}^k}}$, ${{A_{{i_k}}}x_{{i_k}}^{k - \phi _i^k} - {b_{{i_k}}}}$ and all obtained ${\bar w_{\left( {{i_k},j} \right),{i_k}}^k}$ to conduct
\begin{flalign}\label{Asyubariter}
\begin{gathered}
  \bar u_{{i_k}}^k = {\mathcal{P}_ {\mathbb{R}_ + ^q} }(\lambda _{{i_k}}^{k - \phi _{{i_k}}^k} + {\sigma _{{i_k}}}({A_{{i_k}}}x_{{i_k}}^{k - \phi _i^k} \hfill \\
  \qquad\qquad\qquad - {b_{{i_k}}} - \sum\limits_{j \in {\mathcal{N}_{{i_k}}}} {E_{{i_k}j}\bar w_{\left( {{i_k},j} \right),{i_k}}^k} )) \hfill \\
\end{gathered}
\end{flalign}
Next, collecting ${x_{ - {i_k}}^{k - \phi _{ - {i_k}}^k}}$ that determines ${f_{{i_k}}}\left( {{x_{{i_k}}},{x_{ - {i_k}}}} \right)$ and using $\bar u_{{i_k}}^k$, player $i_k \in \mathcal{V}$ finishes the prediction phase after carrying out the  following recursion:
\begin{flalign}\label{Asyxbariter}
\begin{gathered}
  \bar x_{{i_k}}^k \!=\! {\mathcal{P}_{{\Omega _{{i_k}}}}}(x_{{i_k}}^{k - \phi _{{i_k}}^k} \!-\! {\tau _{{i_k}}}({\nabla _{{x_{{i_k}}}}}{f_{{i_k}}}(x_{{i_k}}^{k - \phi _{{i_k}}^k},x_{ - {i_k}}^{k - \phi _{ - {i_k}}^k}) \hfill \\
  \qquad\qquad\quad + A_{{i_k}}^{\rm T}\bar u_{{i_k}}^k)) \hfill \\
\end{gathered}
\end{flalign}
By the obtained prediction phases \eqref{Asywbariter}-\eqref{Asyxbariter}, the active player further drives the next update as follows:
\begin{subequations}\label{Asyxkbariter}
  \begin{equation}\label{Asyxkbaritera}
   x_{{i_k}}^{k + 1} = x_{{i_k}}^k + {\eta _{{i_k}}}(\bar x_{{i_k}}^k - x_{{i_k}}^{k - \phi _{{i_k}}^k}) \qquad\qquad\qquad\,\,\,
  \end{equation}
  \begin{equation}\label{Asyxkbariterb}
  u_{{i_k}}^{k + 1} = u_{{i_k}}^k \!+\! {\eta _{{i_k}}}(\bar u_{{i_k}}^k \!-\! u_{{i_k}}^{k - \phi _{{i_k}}^k} \! + \! {\sigma _{{i_k}}}{A_{{i_k}}}(\bar x_{{i_k}}^k \!-\! x_{{i_k}}^{k - \phi _{{i_k}}^k}))
  \end{equation}
  \begin{equation}\label{Asyxkbariterc}
  \begin{gathered}
  w_{\left( {{i_k},j} \right),{i_k}}^{k + 1} = w_{\left( {{i_k},j} \right),{i_k}}^k + {\eta _{{i_k}}}(\bar w_{\left( {{i_k},j} \right),{i_k}}^k - w_{({i_k},j),{i_k}}^{k - \varphi _{({i_k},j)}^k} \hfill \\
  \qquad\qquad\quad + {\kappa _{({i_k},j)}}{E_{{i_k}j}}(\bar u_{{i_k}}^k - u_{{i_k}}^{k - \phi _{{i_k}}^k})) \hfill \\
\end{gathered} \qquad
  \end{equation}
\end{subequations}

After that, player $i_k$ writes the latest information to its buffer and also sends them to its neighbors. As a result, the detailed implementation of the asynchronous algorithm is summarized in Algorithm \ref{alg:2}.

\begin{algorithm}
	\renewcommand{\algorithmicrequire}{\textbf{Initialization:}}
	\renewcommand{\algorithmicensure}{\textbf{For} $k = 0,{\text{ }}1,{\text{ }} \cdots $
                                      \textbf{do:}\qquad\qquad\qquad\qquad\qquad\qquad\qquad}
	\caption{\textbf{:} ASY\_DPDP\_GNE}
	\label{alg:2}
	\begin{algorithmic}[1]
		\REQUIRE {Each player $i \in \mathcal{V}$ initializes $x_i^0 \in {\mathbb{R}^{{n_i}}}$, $u_i^0 \in{\mathbb{R}^{q}}$ and $ w_{\left( {i,j} \right),i}^0 \in {\mathbb{R}^{q}}$ for $j \in {\mathcal{N}_i}$. Meanwhile, it has a Poisson clock with rata ${\zeta _i}$, and chooses positive step-sizes $\tau _i$, $\sigma_i$ and ${\omega _{(i,j)}}$, and relaxed factor $\eta_i$ properly.
		\ENSURE
		\STATE \textbf{Activation phase:} At iteration $k$, a player $i_k \in \mathcal{V}$ is activated by its local clock with fixed probability $p_i $.

		\STATE \textbf{Receiving phase:}  The active player $i_k$ keeps receiving possibly outdated information ${w_{( {{i_k},j} ),j}^{k - \varphi _{( {{i_k},j})}^k}}$ and $E_{j{i_k}}{u_j^{k - \phi _j^k}}$ from local buffers of each neighbor $j \in {\mathcal{N}_{i_k}}$, as well as $x_j^{k - \phi _j^k}$ that depends on ${f_{{i_k}}}$ from the interference graph. Then, $i_k$ stores these data into its local buffer, and reads them to conduct local computations.

         \STATE \textbf{Prediction phase:} The active player $i_k$ computes $\bar w_{\left( {{i_k},j} \right),{i_k}}^k$, $\bar u_{{i_k}}^k$ and $\bar x_{{i_k}}^k$ as \eqref{Asywbariter}, \eqref{Asyubariter} and \eqref{Asyxbariter}, respectively.

         \STATE \textbf{Update phase:} The active player $i_k$ computes $x_{{i_k}}^{k + 1}$, $u_{{i_k}}^{k + 1}$ and $w_{({i_k},j),{i_k}}^{k + 1}$ as \eqref{Asyxkbaritera}, \eqref{Asyxkbariterb} and \eqref{Asyxkbariterc}, respectively.

         \STATE \textbf{Transmission phase:} The active player $i_k$ writes the latest information $x_{{i_k}}^{k + 1}$, $u_{{i_k}}^{k + 1}$ and $w_{({i_k},j),{i_k}}^{k + 1}$ to its local buffer, and broadcasts the (possibly delayed) data of $x_{i_k}, w_{({i_k},j),{i_k}},{\text{ }}E_{i_kj}u_{{i_k}}$ to its neighbors. Other inactive players only receive the information from neighbors and keep their variables unchanged at previous values.
}
	\end{algorithmic}
\textbf{End.}
\end{algorithm}

\begin{remark}
The receiving phase and the transmission phase compose the interaction, with which the active player enables to read all the stored data  in order to compute a new update. To be specific, the local buffer of each player stores neighboring data (i.e., $E_{ji}u_j$, $w_{(i,j),j}$ and $x_{-i}$) and its own data (i.e., $E_{ij}u_i$, $w_{(i,j),i}$ and $x_{i}$). The active player receives possibly delayed versions of $E_{ji}u_j$ and $w_{(i,j),j}$ via the communication graph assumed in Assumption \ref{AsCommu}, while receives the outdated version of $x_{-i}$ via the interference graph described in Section \ref{SubSynchr}, and copies all received data to its local buffer in practical. Additionally, delays account for the increase in time index during reading, computation, or transmission phases. Therefore, ASY\_DPDP\_GNE allows one player to compute local iteration by reading its local buffer, and complete the writing phase of its latest information before other players finish computations. In this way, players no longer need to wait for slowest neighbors before the execution of local iterations.
\end{remark}

\begin{remark}
It is inevitable to encounter the communication latency in many applications. In the Cournot market competition, for example, the communication network can be formed by viewing factories as the agents in networks. All factories are responsible for communicating their production decisions to ensure that the storage capacities of purchases are always maintained. In this case, the communication latency occurs once one factory makes its decisions slowly. On the other hand, the similar situation will happen in the demand response management, when users exchange information and schedule their energy consumption to satisfy the total demand. Benefiting from the introduced delayed models, the proposed asynchrony system is capable of addressing the above latency issue. To be specific, each factory or user receives outdated information from neighbors and stores them in its local buffer. Then it simply reads ${x_{{i}}^{k - \phi _i^k}}$, ${u_{{i}}^{k - \phi _{{i}}^k}}$, ${\{ {{E_{ji}}u_j^{k - \phi _j^k},w_{\left( {i,j} \right),j}^{k - \varphi _{\left( {i,j} \right)}^k}} \}_{j \in {\mathcal{N}_i}}}$, ${x_{ - {i}}^{k - \phi _{ - {i}}^k}}$ in its buffer to update its latest response.
\end{remark}

\section{Convergence Results} \label{SecConverg}

%

\subsection{Convergence of DPDP\_GNE}
 Let $U = {\rm{col}}\{ x,u,w\} $, ${\delta _\Omega }(x) = \sum\nolimits_{i = 1}^m {{\delta _{{\Omega _i}}}({x_i})}$, ${\delta _{\mathbb{R}_ + ^{mq}}}(u) = \sum\nolimits_{i = 1}^m {{\delta _{\mathbb{R}_ + ^q}}({u_i})} $ and $\delta _\mathcal{C}^ \star (w) = \sum\nolimits_{i = 1}^m {\sum\nolimits_{(i,j) \in \mathcal{E}} {\delta _{{\mathcal{C}_{(i,j)}}}^ \star ({w_{(i,j)}})} } $, where $\Omega  = \prod\nolimits_{i = 1}^m {{\Omega _i}}$ and $\mathcal{C} = \prod\nolimits_{(i,j) \in \mathcal{E}} {{\mathcal{C}_{(i,j)}}} $. Next, we define operators
\begin{flalign} \label{Monotone}
\begin{gathered}
  {T_A}:U \to (\partial {\delta _\Omega }(x),\partial {\delta _{{\mathbb{R}_ + }}}(u),\partial \delta _\mathcal{C}^ \star (w)) \hfill \\
  {T_M}:U \to ({A^{\text{T}}}u, - Ax + {\Pi^{\text{T}}}w, - \Pi) \hfill \\
  {T_C}:U \to (F(x),b,0) \hfill \\
\end{gathered}
\end{flalign}
where $A = {\text{blkdiag}}\{ {A_1}, \cdots ,{A_m}\}  \in {\mathbb{R}^{mq \times mq}}$, $b = {\text{col}}\{ {b_1}, \cdots ,{b_m}\}  \in {\mathbb{R}^{mq}}$, and $\Pi \in {\mathbb{R}^{2\left| \mathcal{E} \right|q \times 2\left| \mathcal{E} \right|q}}$ stacked by all ${\Pi_{(i,j)}}$. Recalling the solution ${U^ * } = {\text{col}}\{ {x^ * },{u^ * },{w^ * }\} $, the condition system \eqref{DisMNcondit} can be compactly expressed as ${U^ * } \in {\rm{zer(}}{T_A} + {T_M} + {T_C}{\rm{)}}$.

\begin{definition} \label{defStpsizeMart}
For each edge $\left( {i,j} \right) \in \mathcal{E}$, define the edge-weight matrix $W = {\rm blkdiag}\{ {{{\left( {{\kappa _{\left( {i,j} \right)}}I_{2q}} \right)}_{\left( {i,j} \right) \in \mathcal{E}}}} \} \in {\mathbb{R}^{2| \mathcal{E} |q \times 2| \mathcal{E} |q}}$. For the step-sizes ${\sigma _1}, \cdots ,{\sigma _m}$ and ${\tau _1}, \cdots ,{\tau _m}$, define the step-size matrices $\Gamma  = {\rm blkdiag}\left\{ {{\tau _i}I_{n_i}} \right\}_{i = 1}^m \in {\mathbb{R}^{n \times n}}$ and $\Sigma = {\rm blkdiag}\left\{ {{\sigma _i}I_q} \right\}_{i = 1}^m \in {\mathbb{R}^{mq \times mq}}$.
\end{definition}

Following the above definition, the compact expression of of Algorithm \ref{alg:1} takes the form of
\begin{flalign} \label{CompactPred}
\begin{gathered}
  {{\bar w}_k} = {\text{prox}}_{\delta _\mathcal{C}^ \star  }^W({w_k} + W{\Pi}{u_k}) \hfill \\
  {{\bar u}_k} = {\text{prox}}_{{\delta _{\mathbb{R}_ + ^{mq}}}}^\Sigma ({\lambda _k} + \Sigma (A{x_k} - b - {\Pi^{\rm T}}{{\bar w}_k})) \hfill \\
  {{\bar x}_k} = {\text{prox}}_{{\delta _\Omega }}^\Gamma ({x_k} - \Gamma (F({x_k}) + {A^{\rm T}}{{\bar u}_k})) \hfill \\
  {x_{k + 1}} = {{\bar x}_k} \hfill \\
  {u_{k + 1}} = {{\bar u}_k} + \Sigma A({{\bar x}_k} - {x_k}) \hfill \\
  {w_{k + 1}} = {{\bar w}_k} + W \Pi ({{\bar u}_k} - {u_k}) \hfill \\
\end{gathered}
\end{flalign}

\begin{definition} \label{AsumpOpertaot}
Let the matrices $\Gamma$, $\Sigma$, and $W$ in Definition \ref{defStpsizeMart} be positive definite. Then, define the matrix ${T_S} = {\rm{blkdiag}}\{ {\Gamma ^{ - 1}},{\Sigma ^{ - 1}},{W^{ - 1}}\}$, and the separable matrices
\begin{flalign*}
{T_H} = \underbrace {\left[ {\begin{array}{*{20}{c}}
   {{\Gamma ^{ - 1}}} & {{{{A^{\rm T}}} \mathord{\left/
 {\vphantom {{{A^{\rm T}}} 2}} \right.
 \kern-\nulldelimiterspace} 2}} & 0  \\
   {{A \mathord{\left/
 {\vphantom {A 2}} \right.
 \kern-\nulldelimiterspace} 2}} & {{\Sigma ^{ - 1}}} & {{{{\Pi^{\rm T}}} \mathord{\left/
 {\vphantom {{{\Pi^{\rm T}}} 2}} \right.
 \kern-\nulldelimiterspace} 2}}  \\
   0 & {{\Pi \mathord{\left/
 {\vphantom {\Pi 2}} \right.
 \kern-\nulldelimiterspace} 2}} & {{W^{ - 1}}}  \\
 \end{array} } \right]}_{{T_P}} + \underbrace {\frac{1}
{2}\left[ {\begin{array}{*{20}{c}}
   0 & {{A^{\rm T}}} & 0  \\
   { - A} & 0 & {{\Pi^{\rm T}}}  \\
   0 & { - \Pi} & 0  \\
 \end{array} } \right]}_{{T_\mathcal{K}}}
\end{flalign*}
\end{definition}

Given Definitions \ref{defStpsizeMart} and \ref{AsumpOpertaot}, the recursion \eqref{CompactPred} can be cast as
\begin{flalign} \label{OperatortildeTT}
U_{k+1} = TU_k
\end{flalign}
where the operator $T$ is given by
\begin{flalign} \label{OperatorTT}
TU = U + T_S^{ - 1}\left( {{T_H} - {T_M}} \right)\left( {\bar U - U} \right)
\end{flalign}
with $\bar U = {\left( {{T_H} + {T_A}} \right)^{ - 1}}\left( {{T_H} - {T_M} - {T_C}} \right)U$. Next, the relationship between ${\rm{fix(}}T{\text{)}}$ and ${\rm{zer(}}{T_A} + {T_M} + {T_C}{\text{)}}$ is discussed.

Suppose that the matrices $T_S$ and $T_P$ in Definition \ref{AsumpOpertaot} are positive definite. Note that the matrix $T_{\mathcal{K}}$ is maximal monotone because of its skew-symmetry \cite{Bauschke2011Convex}. Considering $\hat U \in {\text{fix(}}T{\text{)}}$, one can derive that $\hat U \in {\text{fix(}}T{\text{)}} \Leftrightarrow T\hat U = \hat U \Leftrightarrow \bar U = \hat U \Leftrightarrow {( {{T_H} + {T_A}} )^{ - 1}}( {{T_H} - {T_M} - {T_C}} ) \hat U = \hat U \Leftrightarrow {T_H} \hat U - {T_M}\hat U - {T_C}\hat U \in {T_H}\hat U + {T_A}\hat U \Leftrightarrow \hat U \in {\text{zer(}}{T_A} + {T_M} + {T_C}{\text{)}}$, where the second equivalence follows from the monotonicity of $T_{\mathcal{K}}$ and the positive definiteness of $T_S$ and $T_P$. Therefore, it concludes that ${\text{fix(}}T{\text{)}} = {\text{zer(}}{T_A} + {T_M} + {T_C}{\text{)}}$.
\begin{remark}
Inspired by the splitting scheme mentioned in \cite{Latafat2017Asymmetric}, the compact operator $T$ is split into several simplified sub-operators that are endowed with the properties of monotonicity (e.g., $T_A$, $T_M$, $ {T_\mathcal{K}}$) or positive definiteness (e.g., $T_P$, $T_S$), which favorably facilitates the later analysis.
\end{remark}

\begin{assumption}\label{SynStepsizes} Consider a set of local constants such that $0 < {\alpha _1}, \cdots, {\alpha _m} < 1$ and let ${\kappa _{\left( {i,j} \right)}} >0$. Each player $i \in \mathcal{V}$ holds ${\sigma _i}>0$ such that ${\sigma _i} < 9{(1 - {\alpha _i})^2}/(16\sum\nolimits_{j \in {\mathcal{N}_i}} {{\kappa _{\left( {i,j} \right)}}} )$, and ${\tau _i}>0$ such that
\begin{flalign*}
{\tau _i} < \frac{{2{\mu}{{( { 1 - {\alpha _i}} )}^2}}}{{( {1 - \alpha _i} ){{L_F^2} } + 8{\mu}{\sigma _i}( {1 + {\sigma _i}\sum\nolimits_{j \in {{\mathcal N}_i}} {{\kappa _{\left( {i,j} \right)}}} } ){\lambda _{\max }}(A_i^{\rm T}{A_i})}}
\end{flalign*}
\end{assumption}

\begin{lemma}\label{PosiStepsizes}
Let Assumptions \ref{AssumSetconvex}, \ref{AssumpseudoLipfucnti}, \ref{AsCommu} and \ref{SynStepsizes} hold. Define the matrix $\Theta  = {\rm{blkdiag}}\{ {\Theta _1},{\Theta _2},{\Theta _3}\}$, where ${\Theta _1} = {\rm{blkdiag}}\{ {\alpha _i}{I_{{n_i}}}\} _{i = 1}^m \in {\mathbb{R}^{n \times n}}$, ${\Theta _2} = {\rm{blkdiag}}\{ {\alpha _i}{I_q}\} _{i = 1}^m \in {\mathbb{R}^{mq \times mq}}$ and ${\Theta _3} = {\rm{blkdiag}}\{ {\{ {\alpha _i}{I_q},{\alpha _j}{I_q}\} _{(i,j) \in \mathcal{E}}}\}  \in {\mathbb{R}^{2\left| \mathcal{E} \right|q \times 2\left| \mathcal{E} \right|q}}$. Moreover, consider the symmetric matrix
\begin{flalign*}
{T_{\tilde P}} = \left[ {\begin{array}{*{20}{c}}
   {2{\Gamma ^{ - 1}} - \frac{{L_F^2} }{2\mu}I} & { - {A^{\rm T}}} & {{A^{\rm T}}\Sigma {\Pi^{\rm T}}}  \\
   { - A} & {2{\Sigma ^{ - 1}}} & { - {\Pi^{\rm T}}}  \\
   {\Pi\Sigma A} & { - \Pi} & {2{W^{ - 1}}}  \\
 \end{array} } \right]
\end{flalign*}
Then, ${T_{\tilde P}} - \left( {\Theta  + I} \right){T_S} $ is positive definite.
\end{lemma}
\begin{proof}
See Appendix \ref{ProofPosiStepsizes}.
\end{proof}

Based on above analysis, the following lemma reports an important result of the operator $ T$.
\begin{lemma}\label{LemOperTT}
Let Assumptions \ref{AssumSetconvex}, \ref{AssumpseudoLipfucnti}, \ref{AsCommu} and \ref{SynStepsizes} hold. Recall the operator $ T$ in \eqref{OperatortildeTT}. Consider the constant $\gamma  = {( {1 + {\alpha _{\min }}} )^{ - 1}}$ with ${\alpha _{\min }} = \min \{ {{\alpha _1},{\rm{ }}{\alpha _2},{\rm{ }} \cdots ,{\rm{ }}{\alpha _m}} \}$. Then, $ T$ is $\gamma$-averaged, i.e., for any $U$ and $Z$, it holds that
\begin{flalign} \label{tildeTAverage}
\begin{array}{l}
 \quad \| { TU -  TZ} \|_{{T_S}}^2 \\
  \le \| {U \!-\! Z} \|_{{T_S}}^2 \!-\! \frac{{1 - \gamma }}{\gamma }\| {( {{\rm{Id}} \!-\!  T} )U \!-\! ( {{\rm{Id}} \!-\!  T} )Z} \|_{{T_S}}^2 \\
 \end{array}
\end{flalign}
\end{lemma}
\begin{proof}
See Appendix \ref{ProofLemOperTT}.
\end{proof}

Lemma \ref{LemOperTT} indicates that $T$ is $\gamma$-averaged \cite[Proposition 4.35]{Bauschke2011Convex}. Then, the convergence of DPDP\_GNE follows directly from Theorem \ref{Convergence}.

\begin{theorem}\label{Convergence}
Let Assumptions \ref{AssumSetconvex}, \ref{AssumpseudoLipfucnti}, \ref{AsCommu} and \ref{SynStepsizes} hold. Consider ${U^*} = {\rm{col}}\{ {x^*},{u^*},{w^*}\}  \in {\rm{fix}}{\text{(}}T{\text{)}}$. Then, the sequence ${\left\{ {{U_k}} \right\}_{k \in \mathbb{N}}}$, generated from \eqref{OperatortildeTT}, is quasi-Fej{\'{e}}r monotone, and converges to ${U^ * }$. Moreover, $x^*$ is a solution of $\rm{GNG}$($\Xi$).
\end{theorem}
\begin{proof}
Substituting $U$ and $Z$ with $U_k$ and $U^*$ in \eqref{tildeTAverage} obtains that $\| {{U_{k + 1}} - {U^ * }} \|_{{T_S}}^2 \leq \| {{U_k} - {U^ * }} \|_{{T_S}}^2 - \frac{{1 - \gamma }}{\gamma }\| {{U_k} - {U_{k + 1}}} \|_{{T_S}}^2$. Hence, ${\left\{ {{U_k}} \right\}_{k \in \mathbb{N}}}$ is quasi-Fej{\'{e}}r monotone \cite[Definition 5.32]{Bauschke2011Convex}, and its convergence holds from \cite[Theorem 5.33]{Bauschke2011Convex}. Therefore, $U_k$ converging to the fixed point $ {U^ * }$ of $ T$ means that ${x_k} \to {x^ * }$ and ${u_k} \to {u^ * }$ as $k \to  + \infty $. ${\text{fix(}}T{\text{)}} = {\text{zer(}}{T_A} + {T_M} + {T_C}{\text{)}}$ yields that the pair $({x^*},{u^*})$ meets the condition \eqref{DisMNcondit}. Therefore, we directly apply the result of Lemma \ref{DistriMN} to obtain that ${x^ * } \in {\rm{GNE}}$($\Xi$).
\end{proof}

Theorem \ref{Convergence} shows that $||{U_k} - {U^ * }||_{{T_S}}^2$ is bounded and contractive, so the sequence ${\{ ||{U_k} - {U_{k + 1}}||_{{T_S}}^2\} _{k \in \mathbb{N}}}$ is monotonically nonincreasing \cite{Yan2018A}. Therefore, it follows from \cite[Theorem 1]{Yan2018A} that $||{U_k} - {U_{k + 1}}||_{{T_S}}^2 \le \frac{\gamma }{{1 - \gamma }}\frac{{||{U_0} - {U^ * }||_{{T_S}}^2}}{{k + 1}}$ and ${\{ ||{U_k} - {U_{k + 1}}||_{{T_S}}^2\} _{k \in \mathbb{N}}}$ has the $o(1/(k + 1))$ convergence rate. This is a slightly faster result than the $O(1/(k + 1))$ convergence rate in \cite{Yi2019splitting,Pavel2020Distributed,Peng2020Asynchronous}. Other theoretical insight is left to Subsection \ref{DiscTheore}.

\subsection{Convergence of ASY\_DPDP\_GNE}
In this subsection, the convergence of ASY\_DPDP\_GNE is studied under Assumptions \ref{AssumSetconvex}-\ref{SynStepsizes}, through exploiting the $\gamma$-averaged operator $T$. Define the operator $
{T_i}:U \to {\text{col}}\{ 0, \cdots ,{(TU)_i}, \cdots ,0\}$, where $(TU)_i$ denotes the components of the $i$ player, i.e., ${x_i}$, ${u_i}$ and ${w_{\left( {i,j} \right),i}}$ for all $j \in {\mathcal{N}_i}$. Set the operators $R = {\text{Id}} - T,{\text{ }}{R_i} = ({\text{Id}} - T)_i$, which can be easily checked that $RU = \sum\nolimits_{i = 1}^m {{R_i}U}$. Due to the delays, let ${\phi _k}$ and ${\varphi _k}$ be the delay index vectors, and define the delayed variables vectors at $k$ iteration as: ${x_{k - {\phi _k}}} = {\text{col}}\{ x_1^{k - \phi _1^k}, \cdots ,x_m^{k - \phi _m^k}\}$, ${u_{k - {\phi _k}}} = {\text{col}}\{ u_1^{k - \phi _1^k}, \cdots ,u_m^{k - \phi _m^k}\}$ and ${w_{k - {\varphi _k}}} = {\text{col}}{\{ w_{\left( {i,j} \right)}^{k - \varphi _{\left( {i,j} \right)}^k}\} _{\left( {i,j} \right) \in \mathcal{E}}}$, where $w_{\left( {i,j} \right)}^{k - \varphi _{\left( {i,j} \right)}^k} = {\text{col}}\{ w_{\left( {i,j} \right),i}^{k - \varphi _{\left( {i,j} \right)}^k},w_{\left( {i,j} \right),j}^{k - \varphi _{\left( {i,j} \right)}^k}\}$. Furthermore, consider the vector ${{\overset{\lower0.5em\hbox{$\smash{\scriptscriptstyle\frown}$}}{U} }_k} = {\text{col}}\{ {x_{k - {\phi _k}}},{u_{k - {\phi _k}}},{w_{k - {\varphi _k}}}\}$. Since one player $i_k$ is activated at each iteration, the delayed components corresponding to $i_k$ (i.e., $x_{{i_k}}^{k - \phi _{{i_k}}^k}$, $u_{{i_k}}^{k - \phi _{{i_k}}^k}$ and $w_{\left( {{i_k},j} \right),{i_k}}^{k - \varphi _{\left( {{i_k},j} \right)}^k}$ for $j \in {\mathcal{N}_{i_k}}$) take part in the next iteration. Therefore, following the compact form of the synchronous setting \eqref{OperatortildeTT} and $T_i$, we equivalently formulate the asynchronous distributed iterations as
\begin{flalign} \label{asyniter}
\begin{gathered}
  {U_{k + 1}} = {U_k} - {\eta _{{i_k}}}{R_{{i_k}}}{{\overset{\lower0.5em\hbox{$\smash{\scriptscriptstyle\frown}$}}{U} }_k} \hfill \\
\end{gathered}
\end{flalign}
Next, let ${\mathcal{Z}_k} = \hat \sigma ({U_0},{U_1},{{\overset{\lower0.5em\hbox{$\smash{\scriptscriptstyle\frown}$}}{U} }_1}, \cdots ,{U_k},{{\overset{\lower0.5em\hbox{$\smash{\scriptscriptstyle\frown}$}}{U} }_k})$ denote the smallest $\hat \sigma$-algebra generated by ${U_0},{U_1},{{\overset{\lower0.5em\hbox{$\smash{\scriptscriptstyle\frown}$}}{U} }_1}, \cdots ,{U_k},{{\overset{\lower0.5em\hbox{$\smash{\scriptscriptstyle\frown}$}}{U} }_k}$.

\begin{assumption}\label{AssumpProbibility}
For any $k > 0$, let $i_k$ be the index of the player that is responsible for the $k$-th completed update. Assume that $i_k$ is a random variable that is independent of ${i_1}, \cdots ,{i_{k - 1}}$ as ${\rm{Prob}}( {{i_k} = i})  = {p_i}$.
\end{assumption}

\begin{remark} \label{DiscuProb}
Since each player $i$ follows a Poisson process with parameter ${\zeta _i}$ and only one agent is active at each iteration of Algorithm \ref{alg:2}, any computation occurring at agent $i$ is instant. This case results in a feasible probability selection, i.e., ${p_i} = \frac{{{\zeta _i}}}{{\sum\nolimits_{i = 1}^m {{\zeta _i}} }}$, implying $\sum\nolimits_{i = 1}^m {{p_i}}  = 1$. Moreover, the minimum component probability can be given by $p_{\min} = \frac{\min \{ {\zeta _1}, \cdots, {\zeta _m}\} }{{\sum\nolimits_{i = 1}^m {{\zeta _i}} }}$.
\end{remark}

For ease of analysis, we set a useful synchronous update
\begin{flalign} \label{Fuzhuite}
{{\tilde U}_{k + 1}} = {U_k} - \eta {R} {{\overset{\lower0.5em\hbox{$\smash{\scriptscriptstyle\frown}$}}{U} }_k}
\end{flalign}
Under Assumptions \ref{AssumSetconvex}-\ref{AssumpProbibility}, the relation between ${U_k}$ and ${{\overset{\lower0.5em\hbox{$\smash{\scriptscriptstyle\frown}$}}{U} }_k}$ is established in the following lemma, which appears in \cite{Zhimin2016ARock}.
\begin{lemma} \label{OIYRE}
Let Assumptions \ref{AssumSetconvex}-\ref{AssumpProbibility} hold. For the index set $J\left( k \right) \subseteq \left\{ {k - 1, \ldots ,k - \varepsilon } \right\}$, it holds that ${{\overset{\lower0.5em\hbox{$\smash{\scriptscriptstyle\frown}$}}{U} }_k} = {U_k} + \sum\nolimits_{d \in J\left( k \right)} {\left( {{U_d} - {U_{d + 1}}} \right)}$.
\end{lemma}

Based on the above analysis, we derive an explicit upper bound on the expected distance between $U_{k+1}$ and $U^*$.

\begin{lemma}\label{ConditionQwang}
Let Assumptions \ref{AssumSetconvex}-\ref{AssumpProbibility} hold. Consider the sequence ${\left\{ {{U_k}} \right\}_{k \in \mathbb{N}}}$ generated from \eqref{asyniter}. Set the constant ${p_{\min }} = \min \left\{ {{p_1},{\rm{ }}{p_2},{\rm{ }} \cdots ,{\rm{ }}{p_m}} \right\}$ and the local relaxed factor ${\eta _i} = \eta /(m{p_i})$. Then, for any ${U^*} \in {\rm{zer}}{\text{(}}{T_A} + {T_M} + {T_C}{\text{)}}$ and a constant $c>0$, it holds that
\begin{flalign} \label{rtoq}
\begin{gathered}
  \quad \mathbb{E}\left( {\left\| {{U_{k + 1}} - {U^*}} \right\|_{{T_S}}^2|{\mathcal{Z}_k}} \right) \hfill \\
   \leq \left\| {{U_k} - {U^*}} \right\|_{{T_S}}^2 + \frac{c}{m}\sum\limits_{d \in J\left( k \right)} {\left\| {{U_d} - {U_{d + 1}}} \right\|_{{T_S}}^2}  \hfill \\
  \quad  - \frac{1}{m}\left( {\frac{1}{\eta } - \frac{{\kappa \left( {{T_S}} \right)}}{{m{p_{\min }}}} - \frac{{\left| {J\left( k \right)} \right|}}{c}} \right)\left\| {{{\tilde U}_{k + 1}} - {U_k}} \right\|_{{T_S}}^2 \hfill \\
\end{gathered}
\end{flalign}
\end{lemma}
\begin{proof}
See Appendix \ref{ProfCondition}.
\end{proof}

Note that it is not able to ensure that $\mathbb{E}( {\| {{U_{k + 1}} - {U^*}} \|_{{T_S}}^2|{\mathcal{Z}_k}} )$ $\leq \| {{U_k} - {U^*}} \|_{{T_S}}^2$ due to the existence of the positive asynchronicity term $\sum\nolimits_{d \in J\left( k \right)} {\left\| {{U_d} - {U_{d + 1}}} \right\|_{{T_S}}^2} $ caused by the delay. It is therefore necessary to introduce a variable containing involved delayed vectors. Here, we define the stacked vector ${\mathbf{Z}} = {\text{col}}\left\{ {{Z_l}} \right\}_{l = 0}^\varepsilon$, and an operator $\Phi :{\mathbf{Z}} \to {\text{col}}\{ {\tilde Z_0},{\tilde Z_1}, \cdots ,{\tilde Z_\varepsilon }\}$, where its components are involved
\begin{flalign*}
\begin{gathered}
  {{\tilde Z}_0} = {T_S}{Z_0} + \varepsilon \sqrt {\frac{{{p_{\min }}}}
{{\kappa \left( {{T_S}} \right)}}} {T_S}({Z_0} - {Z_1}) \hfill \\
  {{\tilde Z}_l} = \sqrt {\frac{{{p_{\min }}}}
{{\kappa \left( {{T_S}} \right)}}} {T_S}((l - \varepsilon  - 1){Z_{l - 1}} + (2\varepsilon  - 2l + 1){Z_l} \hfill \\
  \qquad\quad + (l - \varepsilon ){Z_{l + 1}}),{\text{ }}l = 1, \cdots ,\varepsilon - 1 \hfill \\
  {{\tilde Z}_\varepsilon } = \sqrt {\frac{{{p_{\min }}}}
{{\kappa \left( {{T_S}} \right)}}} {T_S}({Z_\varepsilon } - {Z_{\varepsilon  - 1}}) \hfill \\
\end{gathered}
\end{flalign*}
To this end, we recombine the vectors ${{\mathbf{U}}_k} = {\text{col}}\left\{ {{U_{k - l}}} \right\}_{l = 0}^\varepsilon$ and ${{\mathbf{U}}^ * } = {1_{\varepsilon  + 1}} \otimes {U^ * }$. Then, we have that
\begin{flalign}  \label{QiWangHeitiU}
\begin{gathered}
  \quad \left\| {{{\mathbf{U}}_k}} \right\|_\Phi ^2 = \left\| {{U_k} - {U^ * }} \right\|_{{T_S}}^2 + \sqrt {\frac{{{p_{\min }}}}{{\kappa \left( {{T_S}} \right)}}} \hfill \\
  \qquad\quad\quad\quad\,  \times \sum\limits_{d = k - \varepsilon }^{k - 1} {\left( {d - \left( {k \!-\! \varepsilon } \right) + 1} \right)\left\| {{U_d} \!-\! {U_{d + 1}}} \right\|_{{T_S}}^2}  \hfill \\
\end{gathered}
\end{flalign}

\begin{theorem}\label{Fundainequa}
Let Assumptions \ref{AssumSetconvex}-\ref{AssumpProbibility} hold. Consider the constant $\beta  = {\eta ^{ - 1}} - 2\varepsilon \sqrt {\kappa \left( {{T_S}} \right)/{p_{\min }}} /m - \kappa \left( {{T_S}} \right)/(m{p_{\min }}) > 0$. Then, it holds that
\begin{flalign} \label{HeiUkZ}
\mathbb{E}( {\left\| {{{\mathbf{U}}_{k + 1}} \!\!-\!\! {{\mathbf{U}}^ * }} \right\|_\Phi ^2|{\mathcal{Z}_k}} ) \!\!\leq \!\| {{{\mathbf{U}}_k} \!-\!\! {{\mathbf{U}}^ * }} \|_\Phi ^2 \!\!-\!\! \frac{\beta }
{m}\| {{{\tilde U}_{k + 1}} \!\!-\!\! {U_k}} \|_{{T_S}}^2
\end{flalign}
Furthermore, there exists $\tilde \Omega  \in \mathcal{F}$ such that ${\rm{Prob}}(\tilde \Omega ) = 1$, and for every $\xi  \in \tilde \Omega $, the sequence ${\{ {\| {{{\mathbf{U}}_{k }} - {{\mathbf{U}}^ * }} \|_\Phi ^2} \}_{k \in \mathbb{N}}}$ converges.
\end{theorem}
\begin{proof}
See Appendix \ref{ProfFundainequa}.
\end{proof}

\begin{theorem}\label{AsyAlgorithmConver}
Let Assumptions \ref{AssumSetconvex}-\ref{AssumpProbibility} hold. Consider the sequence ${\left\{ {{U_k}} \right\}_{k \in \mathbb{N}}}$ generated from \eqref{asyniter}. Then, ${\left\{ {{U_k}} \right\}_{k \in \mathbb{N}}}$ converges to ${U^*}= {\rm{col}}\{ {x^ * },{u^ * },{w^ * }\} \in {\rm{fix}}{\text{(}}T{\text{)}}$. Moreover, $x^*$ is a solution to the VI problem \eqref{VIproblem}, i.e, ${x^ * } \in {\rm{SOL}}\left( {\mathcal{X},F} \right)$.
\end{theorem}
\begin{proof}
Since ${\{ {\| {{{\mathbf{U}}_{k }} - {{\mathbf{U}}^ * }} \|_\Phi ^2} \}_{k \in \mathbb{N}}}$ converges through Theorem \ref{Fundainequa}, it follows from Opial's Lemma \cite{Opial1967Weak} that the sequence ${\left\{ {{U_k}} \right\}_{k \in \mathbb{N}}}$ also converges to the fixed point of the operator $T$, i.e., $U_k \to U^*$, which means that $x_k \to x^*$ and $u_k \to u^*$. On the other hand, the pair $({x^*},{u^*})$ satisfies the condition system \eqref{DisMNcondit} because ${\text{fix(}}T{\text{)}} = {\text{zer(}}{T_A} + {T_M} + {T_C}{\text{)}}$. Therefore, by Lemma \ref{DistriMN}, it holds that ${x^ * } \in {\rm{SOL}}\left( {\mathcal{X},F} \right)$.
\end{proof}

\begin{remark} \label{RemaImpac}
The i.i.d process assumption for each player is a premise for direct evaluation of conditional expectations (i.e., Lemma \ref{ConditionQwang} and Theorem \ref{Fundainequa}). From Theorem \ref{Fundainequa}, both $\varepsilon$ and ${p_{\min }}$ have impacts on the convergence performance of Algorithm \ref{alg:2}. Specifically, both the smaller value of $\varepsilon$ and bigger value of ${p_{\min }}$ give greater value of $\beta$. This further leads to an increased decrease in ${\{ {\| {{{\mathbf{U}}_{k }} - {{\mathbf{U}}^ * }} \|_\Phi ^2} \}_{k \in \mathbb{N}}}$, thereby providing better performance of ASY\_DPDP\_GNE. As presented in Remark \ref{DiscuProb}, the local Poisson rate might determine the parameter ${p_{\min }}$, and thus imposes an indirect effect on the convergence behavior of ASY\_DPDP\_GNE.
\end{remark}


\subsection{Discussions}\label{DiscTheore}

Due to page limitation, this section briefly discusses the advantages of our algorithm frameworks and analysis approach compared with several well-know schemes.

\emph{ 1) Uncoordinated step-sizes}: Both DPDP\_GNE and ASY\_DPDP\_GNE adopt the uncoordinated and independent fixed step-sizes and share the same selection conditions (see Assumption \ref{SynStepsizes}). Such design is not uncommon in \cite{Yi2019splitting,Franci2021Distributed,Zhaojian2021Energy}. Although we may also adopt decaying step-sizes or global step-sizes in our algorithms, for potentially improving the rate of convergence in practice, we opt not to. The main reason is that in the context of multiplayer systems such choice would require global coordination, that is contradictory to our objective of devising distributed GNE algorithms. As a positive side-effect, the use of uncoordinated and fixed step-sizes can simplify the convergence analysis compared with \cite{Yipeng2021Distributed} employing decaying one. Besides, uncoordinated step-sizes always have less restrictive (or less conservative) selections than global step-sizes \cite{Peng2020Asynchronous}.

\emph{ 2) New distributed extensions}: With proper adjustments and modifications for algorithm frameworks, the extension of new forward-backward-like algorithms is significantly derived. Specifically, recalling the synchronous framework \eqref{Ppredphras}-\eqref{Ppredphrasbarwiji}, we substitute $\bar w_{\left( {i,j} \right),i}$ for \eqref{Ppredphrasb} with $2\bar w_{\left( {i,j} \right),i}-w_{\left( {i,j} \right),i}$, and $\bar u _i$ for \eqref{Ppredphrasc} with $2\bar u _i-u _i$. The explicit forms of $\bar w_{\left( {i,j} \right),i}$, $\bar u _i$ and $\bar x _i$ can be easily derived from \eqref{Ppredphrasbarux}-\eqref{Ppredphrasbarwiji}. Next, removing ${\sigma _i}{A_i}({{\bar x}_i - x_i})$ and $\kappa _{\left( {i,j} \right)}{E_{ij}}(\bar u_i -u_i)$ in the correction step, the new synchronous distributed forward-backward algorithm is obtained. On the other hand, recalling the asynchronous framework \eqref{Asywbariter}-\eqref{Asyxkbariter}, we replace $\bar w_{\left( {{i_k},j} \right),{i_k}}^k$ in \eqref{Asyubariter} with $2\bar w_{\left( {{i_k},j} \right),{i_k}}^k - w_{\left( {{i_k},j} \right),{i_k}}^k$, $\bar u_{{i_k}}^k$ from \eqref{Asyxbariter} with $2\bar u_{{i_k}}^k - u_{{i_k}}^k$. Then, removing ${\sigma _{{i_k}}}{A_{{i_k}}}(\bar x_{{i_k}}^k \!-\! x_{{i_k}}^{k - \phi _{{i_k}}^k})$ in \eqref{Asyxkbariterb} and ${E_{{i_k}j}}(\bar u_{{i_k}}^k - u_{{i_k}}^{k - \phi _{{i_k}}^k})$ in \eqref{Asyxkbariterc} yields the new asynchronous distributed forward-backward algorithm. Note that the above new extended algorithms also enjoy distinct attributions as discussed in Remark \ref{ReAdvanEdge}. Besides, each player in both extensions does not require latest auxiliary information from neighbors, compared with recent works \cite{Yi2019splitting,Pavel2020Distributed,Franci2021Distributed}.

\emph{ 3) Improved convergence rates}: These forward-backward-based algorithms \cite{Yi2019splitting,Pavel2020Distributed,Peng2020Asynchronous,Cenedese2019An} can be viewed as our analysis framework \eqref{OperatortildeTT}-\eqref{OperatorTT}. Here we only present key details for \cite[Eq. 10]{Peng2020Asynchronous}. Define ${T_M}:U \to ( {{A^{\rm T}}u , - \tilde Vz - Ax,\tilde Vu } )$, ${T_A}:U \to ( {{N_\Omega }\left( x \right),{N_{\mathbb{R}_ + ^{mr}}}\left( u \right),0} )$, and ${T_C}:U \to ( {F\left( x \right),b,0} )$, where $\tilde V $ stands for incidence matrix and $z$ is edge-based variable \cite{Peng2020Asynchronous}. The fixed-point iteration \cite[Eq. 10]{Peng2020Asynchronous} can be given by \eqref{OperatorTT} with ${T_\mathcal{K}}={T_M}$, ${T_P} = {T_S} $ and ${T_S}:U \to ({\tau ^{ - 1}}x + {A^{\rm T}}u,Ax + {\sigma ^{ - 1}}u + {{\tilde V}^{\rm T}}z,\tilde Vu + {\omega ^{ - 1}}z)$, where $\tau ,\sigma ,\omega $ are positive step-sizes (scalars). Under conditions $\omega  < 1$, $\sigma  < 1/(1 + {\lambda _{\max }}({{\tilde V}^{\rm T}} \tilde V))$, $\tau  < 2\mu (2 - \eta (1 + \alpha ))/(L_F^2 + 2\mu (2 - \eta \left( {1 + \alpha } \right))(1 + {\lambda _{\max }}({A ^{\rm T}}A)))$ with $0<\alpha<1$ and $0<\eta<2/(1+\alpha)$, it can be verified that ${T_{\tilde P}} - \eta \left( {\alpha  + 1} \right){T_S}$ is positive definite via Lemma \ref{PosiStepsizes}, where ${T_{\tilde P}}:U \to (2{\tau ^{ - 1}}x - L_F^2x/(2\mu ) + {A^{\rm T}}u,Ax + 2{\sigma ^{ - 1}}u + {{\tilde V}^{\rm T}}z,\tilde Vu + 2{\omega ^{ - 1}}z)$. Then, the analysis from Lemma \ref{LemOperTT} and Theorem \ref{Convergence} can be used to obtain $\| {{U_k} - {U_{k + 1}}} \|_{{\eta ^{ - 1}}{T_S}}^2 = o(1/(k + 1))$. This generalizes the result in \cite[Lemma 2]{Peng2020Asynchronous} and improves its $O(1/(k + 1))$ convergence rate to $o(1/(k+1))$. In the same way, our analytical framework with appropriate selections of $T_S$, $T_P$ and ${T_\mathcal{K}}$ can also be applied to \cite{Yi2019splitting,Pavel2020Distributed,Cenedese2019An}, yielding or improving their convergence rates to $o(1/(k+1))$.

\section{Performance Evaluation} \label{SecPerfo}

In this section, we validate the performance of ASY\_DPDP\_GNE on the practical example of the Cournot market competition problem introduced in Section \ref{SubPrSce}. Assume that there are ten factories,
one commodity, and four purchasers. The procurement relationship between factories and purchasers is connected as shown in \cite[Fig. 1]{Zheng2023Stochastic}. The parameters ${g_{j,s}}$, ${\rho _{j,s}}$, ${a _{i,s}}$, ${c _{i,s}}$ are randomly draw from intervals $[20,50]$, $[2,3]$, $[0.1,1]$ and $[1,10]$. The production of upper bounds $q_{j,\max }^i$ is uniformly set as $50$, the limited storage capacity ${b_s}$ is respectively defined as $30, 50, 40, 20$. Obviously, for each $i \in \mathcal{W}$, ${\Omega _i} = \prod\nolimits_{j \in {\mathcal{D}_j}} {{\Omega _{i,j}}} $ $A_i$ and $f_i$ satisfy Assumption \ref{AssumSetconvex}, and the pseudo-gradient mapping $F$ satisfies Assumption \ref{AssumpseudoLipfucnti}. Thus, the Cournot market competition problem admits a unique variational GNE. Following from Assumption \ref{SynStepsizes} and Theorem \ref{Fundainequa}, we set the step-sizes of ASY\_DPDP\_GNE as follows: ${\alpha _i} = 0.25$, ${\kappa _{\left( {i,j} \right)}} = \max \left\{ {\left| {{\mathcal{N}_i}} \right|,\left| {{\mathcal{N}_j}} \right|} \right\}$, ${\sigma _i} = 0.5{(1 - {\alpha _i})^2}/\sum\nolimits_{j \in {\mathcal{N}_i}} {{\kappa _{\left( {i,j} \right)}}}$, ${\tau _i} = {(1 - {\alpha _i})^2}/(15(1 - {\alpha _i})+ 16{\sigma _i}(1 + {\sigma _i}\sum\nolimits_{j \in {\mathcal{N}_i}} {{\kappa _{\left( {i,j} \right)}}} ))$ and $\eta  = \left( {m - 1} \right)\cdot{p_{\min }}/(\sqrt {\kappa ({T_S})} (2\varepsilon \sqrt {{p_{\min }}}  + \sqrt {\kappa ({T_S})} ))$. We use the residuals $\sum\nolimits_{i = 1}^m {\| {x_i^k - x_i^ * } \|/\| {x_i^ * } \|}$ and $\sum\nolimits_{i = 1}^m {\| {u_i^k - u_g^ * } \|/m\| {u_g^ * } \|}$ to plot the dynamics trajectories, where $x_i^*$ is computed by a centralized method. To study the performance of ASY\_DPDP\_GNE, we conduct the simulation via the following three aspects.

\emph{ 1) Effects of active proportions}: This case investigates the effects of active proportions on the performance. According to Assumptions \ref{AssumpDelay} and \ref{AssumpProbibility}, we fix the the
upper bound $\varepsilon =5$, and consider three groups of active proportions, where all $p_i$, $i \in \mathcal{W}$, are selected to satisfy $\sum\nolimits_{i = 1}^m {{p_i}}  = 1$, but the minimal probabilities $p_{min}$ are respectively as $0.03$, $0.06$ and $0.1$. Fig. \ref{DiffProp} confirms the algorithm performance under different active proportions, revealing that bigger minimal probability results in better convergence performance.

\emph{ 2) Effects of delays}: This case fixes the probability $p_i = 1/m$ and consider different maximal delays $\varepsilon  \in \{ 5,10,15,20\} $. The simulation results in Fig. \ref{DiffDelay} indicate that the smaller value of $\varepsilon$ leads to faster convergence of ASY\_DPDP\_GNE exhibits, aligning with the discussions in Remark \ref{RemaImpac}.

\begin{figure}[!t]
   \centering
   \subfigure[The magnitude of active proportions.]{\includegraphics[width=1.5in]{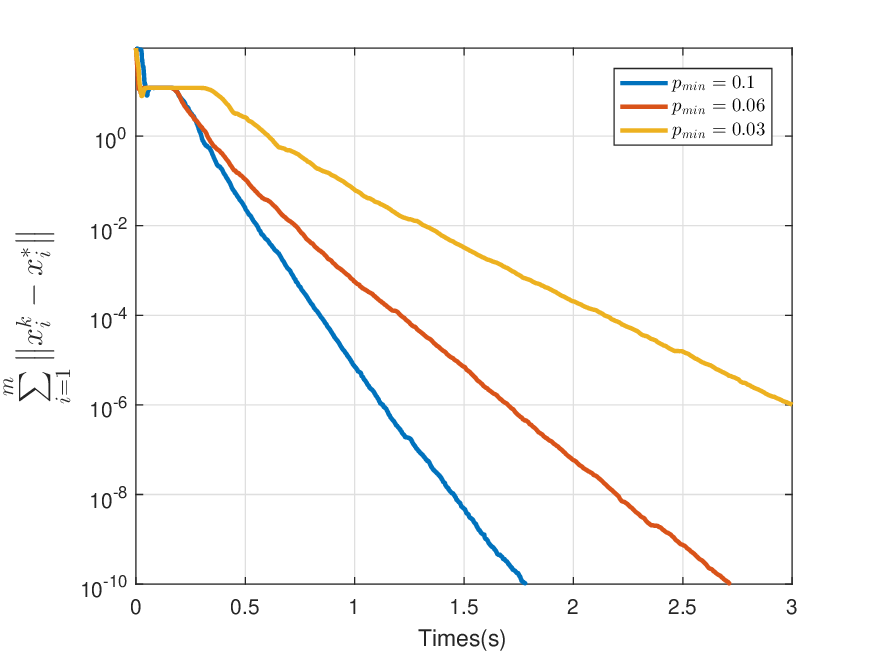} \label{DiffProp}}
   \hfil
   \subfigure[The magnitude of delays.]{\includegraphics[width=1.5in]{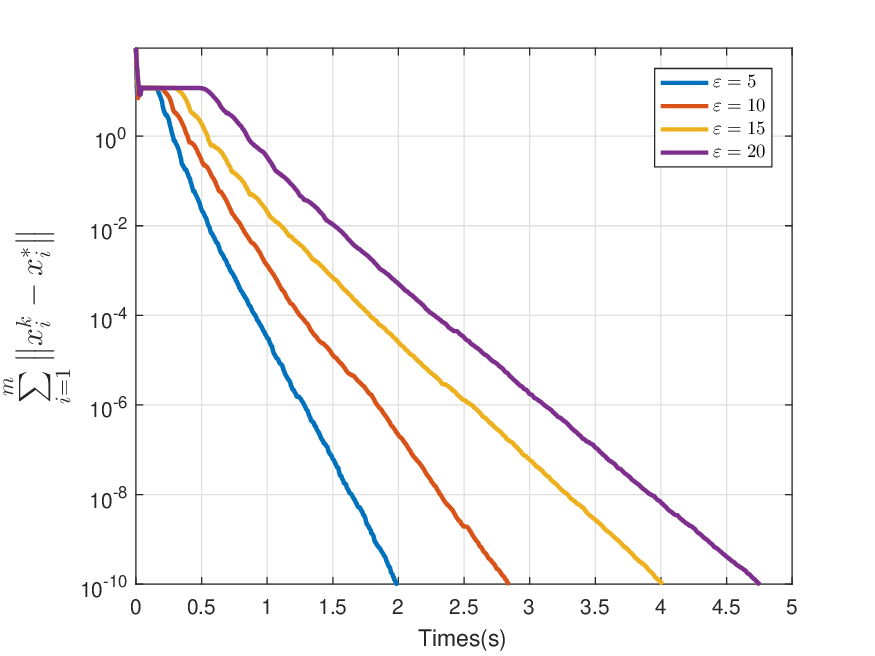} \label{DiffDelay}}
   \caption{Comparison of different considerations.}\label{StateDual}
\end{figure}

\emph{ 3) Comparisons with state-of-the-arts}: In this case, we compare ASY\_DPDP\_GNE with DPDP\_GNE, the Nesterov-based method in \cite{Zhaojian2021Energy,Guannan2020Accelerated} (SGNE for short), ADAGNES as well as its synchronous version (SYDNEY for short) in \cite{Peng2020Asynchronous}. To simulate a practical environment, the computation time of player $i$ is sampled from an exponential distribution $\exp (1/{\nu _i})$, where $\nu _i$ is set as $1 + |\nu |$ with $\nu$ following the standard normal distribution $N(0,1)$. The maximal delay is defined as $10$. The step-sizes of the comparative algorithms are carefully hand-tuned to achieve their best performance under their feasible selection conditions. Fig. \ref{ComWitAlgo} highlights that i) ASY\_DPDP\_GNE is much faster than ADAGNES, potentially  benefitting from the use of uncoordinated step-sizes; ii) ASY\_DPDP\_GNE significantly faster than synchronous counterparts. This is because ASY\_DPDP\_GNE does not require faster players to wait for slower ones, thereby eliminating the waiting time. Moreover, the synchronous DPDP\_GNE outperforms SGNE and SYDNEY, which is attributed to the theoretical results in Theorem \ref{Convergence} and Section \ref{DiscTheore}.

\begin{figure}[!t]
   \centering
   \subfigure[Evolutions of the residuals of dual variables.]{\includegraphics[width=1.5in]{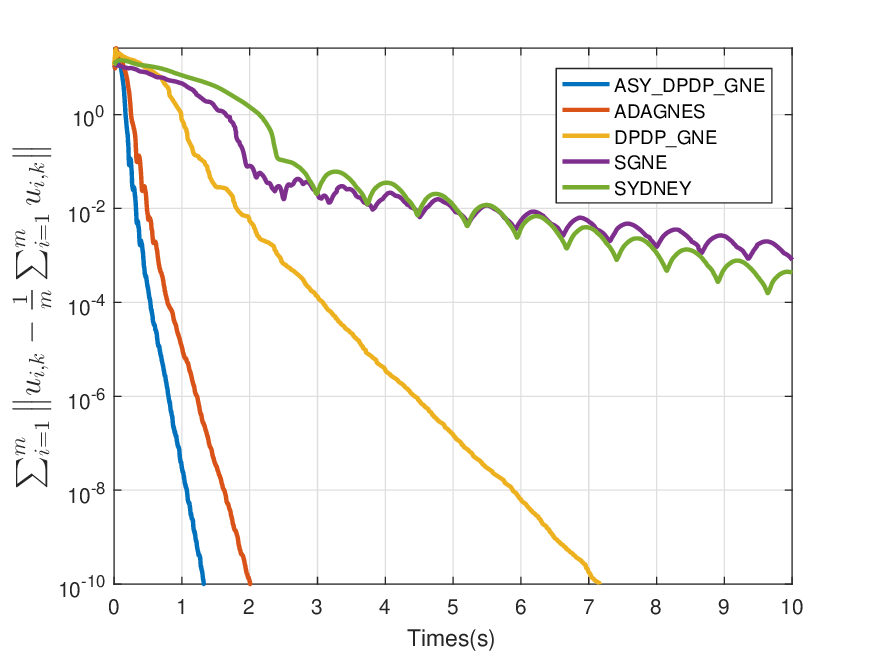} \label{CompDual}}
   \hfil
   \subfigure[Evolutions of the residuals of decision variables.]{\includegraphics[width=1.5in]{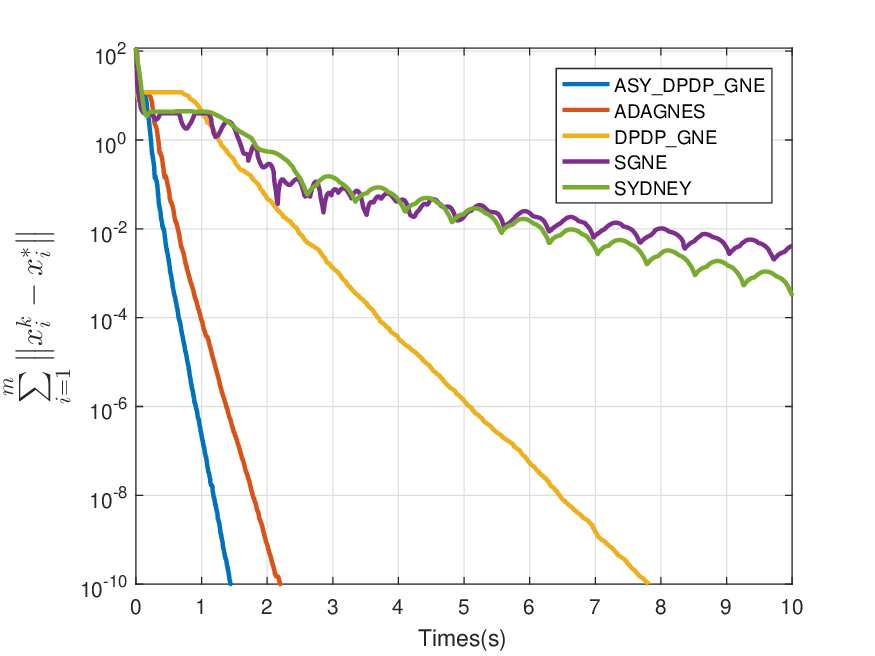} \label{DualPria}}
   \caption{Performance comparison with state-of-the-arts.}\label{ComWitAlgo}
\end{figure}

\section{Conclusion} \label{SecConcl}
This study developed DPDP\_GNE and its asynchronous version ASY\_DPDP\_GNE with delay communications for noncooperative games with coupled inequality constraints. Both proposed algorithms employ a new edge-based mode with distinct advantages compared with other edge-based methods. By means of operator-splitting techniques and the averaged operator theory, the convergence was explicitly derived under mild assumptions. Simulations are conducted to verify the performance of the proposed algorithms. It is necessary to investigate asynchronous distributed algorithms with partial-decision information in future work. In addition, games for the time-varying networks, distributed optimal power flow algorithm for radial communication networks are also potential topic worthy of careful consideration.

\begin{appendices}
\subsection{Proof of Lemma \ref{DistriMN}} \label{ProofDistriMN}
\begin{proof}
Firstly \eqref{DisMNcondita} $ \Rightarrow $ \eqref{VIKKTa} is shown. It follows from \eqref{DisMNconditc} that $w_{(i,j)}^* = {\text{pro}}{{\text{x}}_{{\kappa _{(i,j)}}\delta _{{\mathcal{C}_{(i,j)}}}^ \star }}( {w_{(i,j)}^* + {\kappa _{(i,j)}}{\Pi_{(i,j)}}{u^*}} )$, where ${\kappa _{(i,j)}}$ is a positive stepsize. Using Moreau decomposition \cite{Bauschke2011Convex} yields that ${\Pi_{(i,j)}}{u^*} = {\mathcal{P}_{{\mathcal{C}_{(i,j)}}}}( {\kappa _{(i,j)}^{ - 1}{{\hat w}_{(i,j)}} + {\Pi_{(i,j)}}{u^*}} )$, where the fact that ${\text{pro}}{{\text{x}}_{\kappa _{(i,j)}^{ - 1}{\delta _{{\mathcal{C}_{(i,j)}}}}}}\left( {{\text{Id}}} \right) = {\mathcal{P}_{{\mathcal{C}_{(i,j)}}}}\left( {{\text{Id}}} \right)$ is used. Note that the projection of $({z_1},{z_2}) \in {\mathbb{R}^n} \times {\mathbb{R}^n}$ onto ${{\mathcal{C}_{(i,j)}}}$ is derived as ${\mathcal{P}_{{\mathcal{C}_{(i,j)}}}}:({z_1},{z_2}) \to \frac{1}{2}({z_1} - {z_2},{z_2} - {z_1})$. Recalling $w_{(i,j)}^* = {\text{col}}\{ w_{(i,j),i}^*,w_{(i,j),j}^*\} $, $\Pi_{(i,j)}$ and above results becomes
\begin{flalign*}
\left[ {\begin{array}{*{20}{c}}
   {{E_{ij}}u_i^ *  + {E_{ij}}u_j^ * }  \\
   {{E_{ji}}u_j^ *  + {E_{ij}}u_i^ * }  \\
 \end{array} } \right] = \frac{1}
{{{\kappa _{(i,j)}}}}\left[ {\begin{array}{*{20}{c}}
   {w_{(i,j),i}^* - w_{(i,j),j}^*}  \\
   {w_{(i,j),j}^* - w_{(i,j),i}^*}  \\
 \end{array} } \right]
\end{flalign*}
Multiplying both sides of the above result by $1_{2q}^{\text{T}}$ obtains $w_{\left( {i,j} \right),i}^ *  = w_{\left( {i,j} \right),j}^ *$ and ${\text{ }}{E_{ij}}u_i^ *  + {E_{ji}}u_j^ *  = 0$. Recalling $E_{ij}$ in \eqref{EijEji}, we further have $u_i^* = u_j^* = u_g^*$, which combines with \eqref{DisMNcondita} to reach \eqref{VIKKTa}.

Next, \eqref{DisMNconditb} $ \Rightarrow $ \eqref{VIKKTb} is proved. It follows from $u_i^* = u_j^* = u_g^*$ that $u_1^* =  \cdots  = u_m^* = u_g^*$. Then, it holds that ${N_{\mathbb{R}_ + ^q}}(u_i^*) = {N_{\mathbb{R}_ + ^q}}(u_g^*)$ for each player $i$. Therefore, in \eqref{DisMNconditb}, there exist ${a_1},{\text{ }} \cdots ,{\text{ }}{a_m} \in {N_{\mathbb{R}_ + ^q}}(u_g^ * )$ such that $0 = \sum\nolimits_{i = 1}^m {{a_i}}  - \sum\nolimits_{i = 1}^m {\left( {{A_i}x_i^* - {b_i}} \right)}$, where we use $\sum\nolimits_{i = 1}^m {\sum\nolimits_{j \in {\mathcal{N}_i}} {{E_{ij}}w_{(i,j),i}^*} }  = 0$. Since ${a_i} \in {N_{\mathbb{R}_ + ^q}}(u_g^ * )$, $\sum\nolimits_{i = 1}^m {{a_i}}  \in {N_{\mathbb{R}_ + ^q}}(u_g^ * )$ holds, one can further have that $0 \in {N_{\mathbb{R}_ + ^q}}(u_g^*) - \sum\nolimits_{i = 1}^m {({A_i}x_i^* - {b_i})} $, which coincides with \eqref{VIKKTb}. Therefore, the pair $({x^ * },{u^ * })$ with ${u^ * } = {\text{col}}\{ u_1^*,{\text{ }} \cdots ,{\text{ }}u_m^*\}  = {1_{mq}} \otimes u_g^*$ meets the conditions \eqref{VIKKT}, and ${x^ * } \in {\rm{SOL}}\left( {\mathcal{X},F} \right)$.
\end{proof}

\subsection{Proof of Lemma \ref{PosiStepsizes}} \label{ProofPosiStepsizes}
\begin{proof}
The positive definiteness of ${T_{\tilde P}} - \left( {\Theta  + I} \right){T_S} $ can be equivalently expressed by the following inner product
\begin{flalign} \label{Poticefgd}
\langle {U,( {{T_{\tilde P}} - ( {\Theta  + I} ){T_S}} )U} \rangle  > 0,\forall U = {\text{col}}\{ x,u,w\}
\end{flalign}
Concretely, expanding the left-hand side of \eqref{Poticefgd} yields
\begin{flalign} \label{ProductUTS}
\begin{gathered}
  \quad \left\langle {U,\left( {{T_{\tilde P}} - \left( {\Theta  + I} \right){T_S}} \right)U} \right\rangle  \hfill \\
   = \left\| x \right\|_{\left( {I - {\Theta _1}} \right){\Gamma ^{ - 1}} - \frac{{L_F^2} }{2\mu}I}^2 + \left\| u \right\|_{\left( {I - {\Theta _2}} \right){\Sigma ^{ - 1}}}^2 \hfill \\
 \quad + \left\| w \right\|_{\left( {I - {\Theta _3}} \right){W^{ - 1}}}^2 - 2\left\langle {u,Ax} \right\rangle  \hfill \\
  \quad + 2\left\langle {w,\Pi \Sigma Ax} \right\rangle  - 2\left\langle {w,\Pi u} \right\rangle  \hfill \\
\end{gathered}
\end{flalign}
Under Assumption \ref{SynStepsizes}, these matrices ${\Gamma }$, ${\Sigma }$, ${W}$, $I - {\Theta _1}$, $I - {\Theta _2}$ and $I - {\Theta _3}$ are positive definite. Meanwhile, note that
\begin{flalign*}
\begin{gathered}
   \quad - 2\left\langle {u,Ax} \right\rangle + 2\left\langle {w,\Pi \Sigma Ax} \right\rangle - 2\left\langle {w,\Pi u} \right\rangle \hfill \\
   \ge - \frac{1}{4}{\| {\sqrt {{\Sigma ^{ - 1}}\left( {I - {\Theta _2}} \right)} u} \|^2} - 4{\| {\sqrt {\Sigma {{\left( {I - {\Theta _2}} \right)}^{ - 1}}} Ax} \|^2} \hfill \\
   \quad - \frac{1}{4}{\| {\sqrt {{W^{ - 1}}\left( {I - {\Theta _3}} \right)} w} \|^2} - \frac{3}{4}\left\| w \right\|_{\left( {I - {\Theta _3}} \right){W^{ - 1}}}^2 \hfill \\
   \quad - 4{\| {\sqrt {W{{\left( {I \!-\! {\Theta _3}} \right)}^{ - 1}}} \Pi \Sigma Ax} \|^2} \!-\! \frac{4}{3}{\| {\sqrt {W{{\left( {I \!-\! {\Theta _3}} \right)}^{ - 1}}} \Pi u} \|^2} \hfill \\
\end{gathered}
\end{flalign*}
Plugging the above inequalities into \eqref{ProductUTS}, we have $\langle {U,\left( {{T_{\tilde P}} - \left( {\Theta  + I} \right){T_S}} \right)U} \rangle  \ge \left\| x \right\|_X^2 + \frac{3}
{4}\left\| u \right\|_Y^2$, where $X = (I - {\Theta _1}){\Gamma ^{ - 1}} - \frac{{L_F^2} }{2\mu}I - 4{A^{\rm T}}\Sigma ({(I - {\Theta _2})^{ - 1}}+ {\Pi ^{\rm T}}W{(I - {\Theta _3})^{ - 1}}\Pi \Sigma )A$ and $Y = \frac{3}{4}\left( {I - {\Theta _2}} \right){\Sigma ^{ - 1}} - \frac{4}{3}{\Pi^{\rm T}}W{\left( {I - {\Theta _3}} \right)^{ - 1}}\Pi$. Since $\Gamma ,\Sigma ,W,{\Theta _1},{\Theta _2},{\Theta _3}$ and $A$ have diagonal structures and ${\Pi ^{\rm T}}\Pi  = {\rm{blkdiag}}\{ \sum\nolimits_{j \in {\mathcal{N}_i}} {E_{ij}^{\rm T}{E_{ij}}} \} _{i = 1}^m$, the upper bound of ${\tau _i}$ in Assumption \ref{SynStepsizes} implies that $\left\| x \right\|_X^2 >0$, and the upper bound of ${\sigma _i}$ means that $\left\| u \right\|_Y^2 >0$. As a result, the inequality \eqref{Poticefgd} holds.
\end{proof}

\subsection{Proof of Lemma \ref{LemOperTT}} \label{ProofLemOperTT}
\begin{proof}
For any $U$ and $Z$, it follows from \eqref{OperatorTT} that
\begin{flalign}\label{CompaTildeT}
 \begin{gathered}
  {T_S}( { TU - U} ) = ( {{T_H} - {T_M}} )( {\bar U - U} ) \hfill \\
  {T_S}( { TZ - Z} ) = ( {{T_H} - {T_M}} )( {\bar Z - Z} ) \hfill \\
\end{gathered}
\end{flalign}
where ${\bar U}$ and ${\bar Z}$ are respectively given by
\begin{flalign}\label{TildeUZ}
\begin{array}{l}
 \bar U = {\left( {{T_H} + {T_A}} \right)^{ - 1}}\left( {{T_H} - {T_M} - {T_C}} \right)U \\
 \bar Z = {\left( {{T_H} + {T_A}} \right)^{ - 1}}\left( {{T_H} - {T_M} - {T_C}} \right)Z \\
\end{array}
\end{flalign}
Since $T_M$ and $T_{\mathcal{K}}$ are skew-adjoint, the splitting scheme of $T_H$ in Definition \ref{AsumpOpertaot} gives that $\left\langle {U,\left( {{T_H} - {T_M}} \right)U} \right\rangle  \ge \left\| U \right\|_{{T_P}}^2$. Combining this result and the relationship \eqref{CompaTildeT}, we can derive
\begin{flalign}\label{UkjianUstar}
\begin{gathered}
  \quad \langle {U - Z,{\text{ }}{T_S}( { TU - U} ) - ( { TZ - Z} )} \rangle  \hfill \\
   = \left\langle {\bar U - \bar Z,\left( {{T_H} - {T_M}} \right)\left( {\left( {\bar U - U} \right) - \left( {\bar Z - Z} \right)} \right)} \right\rangle  \hfill \\
   \quad - \left\langle {\left( {\bar U \!-\! U} \right) \!-\! \left( {\bar Z \!-\! Z} \right),\left( {{T_H} \!-\! {T_M}} \right)\left( {\left( {\bar U \!-\! U} \right) \!-\! \left( {\bar Z \!-\! Z} \right)} \right)} \right\rangle  \hfill \\
   \leq  - \left\| {\left( {\bar U - U} \right) - \left( {\bar Z - Z} \right)} \right\|_{{T_P}}^2 \hfill \\
  \quad  + \left\langle {\bar U - \bar Z,{\text{ }}\left( {{T_H} - {T_M}} \right)\left( {\left( {\bar U - U} \right) - \left( {\bar Z - Z} \right)} \right)} \right\rangle  \hfill \\
\end{gathered}
\end{flalign}
Next we analyze the cross terms in \eqref{UkjianUstar}. We use \eqref{TildeUZ} to have
\begin{flalign*}
\begin{gathered}
  \quad \left\langle {\bar U - \bar Z,{\text{ }}\left( {{T_H} - {T_M}} \right)\left( {\left( {\bar U - U} \right) - \left( {\bar Z - Z} \right)} \right)} \right\rangle  \hfill \\
   =  \!-\! \left\langle {\bar U \!-\! \bar Z,\left( {{T_A}\bar U \!-\! {T_A}\bar Z} \right)} \right\rangle \!-\! \left\langle {\bar U \!-\! \bar Z,\left( {{T_M}\bar U \!-\! {T_M}\bar Z} \right)} \right\rangle  \hfill \\
  \quad - \left\langle {\bar U - \bar Z,\left( {{T_C}U - {T_C}Z} \right)} \right\rangle  \hfill \\
   \leq  \!-\! \left\langle {\bar U \!-\! \bar Z,\left( {{T_C}U \!-\! {T_C}Z} \right)} \right\rangle  \hfill \\
\end{gathered}
\end{flalign*}
where the inequality holds from the monotonicity of ${{T_M}}$  and ${{T_C}}$. Assumption \ref{AssumpseudoLipfucnti} indicates that for any $x  ,y \in {\mathbb{R}^n}$
\begin{flalign} \label{CompactLipco}
\left\langle {x - y,F\left( x \right) - F\left( y \right)} \right\rangle  \ge \frac{\mu }{{L_F^2}}{\left\| {F\left( x \right) - F\left( y \right)} \right\|^2}
\end{flalign}
On the other hand, the basic inequality $\langle {x,y} \rangle \leq \frac{1}{2}\left\| x \right\|_V^2 + \frac{1}{2}\left\| y \right\|_{{V^{ - 1}}}^2$ ($V \in {\mathbb{R}^{n \times n}}$ is positive definite) gives that
\begin{flalign}\label{fertfw}
\begin{gathered}
  \quad \left\langle {\left( {x - \bar x} \right) - \left( {y - \bar y} \right),F\left( x \right) - F\left( y \right)} \right\rangle  \hfill \\
   \leq \frac{\mu }{{L_F^2}} \left\| {F\left( x \right) - F\left( y \right)} \right\|^2 + \frac{{L_F^2} }{4\mu}\left\| {\left( {x - \bar x} \right) - \left( {y - \bar y} \right)} \right\|^2 \hfill \\
\end{gathered}
\end{flalign}
Combine \eqref{CompactLipco} with \eqref{fertfw} to have the upper bound of cross term: $\langle {\bar U - \bar Z,( {{T_H} - {T_M}} )( {( {\bar U - U} ) - ( {\bar Z - Z} )} )} \rangle  \leq L_F^2/( {4\mu } ){\| {( {x - \bar x} ) - ( {y - \bar y} )} \|^2}$. Therefore, \eqref{UkjianUstar} reduces to $\langle {U - Z,{\text{ }}{T_S}( { TU - U} ) - ( { TZ - Z} )} \rangle  \leq  - \| {( { TU - U} ) - ( {TZ - Z} )} \|_{{T_{\bar P}}}^2$, where the last term is derived from \eqref{CompaTildeT}. Note that $\langle {U - Z,{T_S}(TU - U) - (TZ - Z)} \rangle  = \| {TU - TZ} \|_{{T_S}}^2 - \| {U - Z} \|_{{T_S}}^2 - \| {({\rm{Id}} - T)U - ({\rm{Id}} - T)Z} \|_{{T_S}}^2$. Consequently, combining the above results and Lemma \ref{PosiStepsizes}, we can obtain that $\| { TU -  TZ} \|_{{T_S}}^2 - \| {U - Z} \|_{{T_S}}^2 \leq  - \| {( { TU - U} ) - ( { TZ - Z} )} \|_{\Theta {T_S}}^2$. Since $\Theta$ is a block diagonal matrix, we take its minimum eigenvalue as ${\alpha _{\min }} = \min \left\{ {{\alpha _1},{\text{ }}{\alpha _2},{\text{ }} \cdots ,{\text{ }}{\alpha _m}} \right\}$ and set $\left( {1 - \gamma } \right)/\gamma  = {\alpha _{\min }}$. Therefore, the result \eqref{tildeTAverage} holds.
\end{proof}

\subsection{Proof of Lemma \ref{ConditionQwang}} \label{ProfCondition}
\begin{proof}
Under Assumption \ref{AssumpProbibility}, we use \eqref{asyniter} and ${\eta _i} = \eta /(m{p_i})$ to derive the following conditional expectation:
\begin{flalign} \label{ConditrTUbar}
\begin{gathered}
  \quad \mathbb{E}(||{U_{k + 1}} - {U^*}||_{{T_S}}^2|{\mathcal{Z}_k}) \hfill \\
    =  \mathbb{E}(||{U_k} - {\eta _{{i_k}}}{{\overset{\lower0.5em\hbox{$\smash{\scriptscriptstyle\frown}$}}{R} }_{{i_k}}}{{\overset{\lower0.5em\hbox{$\smash{\scriptscriptstyle\frown}$}}{U} }_k} - {U^*}||_{{T_S}}^2|{\mathcal{Z}_k}) \hfill \\
  =  ||{U_k} - {U^*}||_{{T_S}}^2 + \frac{{{\eta ^2}}}
{{{m^2}}}\sum\limits_{i = 1}^m {\frac{1}{{{p_i}}}||{{\overset{\lower0.5em\hbox{$\smash{\scriptscriptstyle\frown}$}}{R} }_i}{{\overset{\lower0.5em\hbox{$\smash{\scriptscriptstyle\frown}$}}{U} }_k}||_{{T_S}}^2}  \hfill \\
  \quad + \frac{{2\eta }}{m}{\langle {{U^*} - {U_k},\overset{\lower0.5em\hbox{$\smash{\scriptscriptstyle\frown}$}}{R} {{\overset{\lower0.5em\hbox{$\smash{\scriptscriptstyle\frown}$}}{U} }_k}} \rangle _{{T_S}}} \hfill \\
\end{gathered}
\end{flalign}
For the second term in the last equality of \eqref{ConditrTUbar}, it is derived as
\begin{flalign} \label{HuRiUU}
\begin{gathered}
  \sum\limits_{i = 1}^m {\frac{1}
{{{p_i}}}||{{\overset{\lower0.5em\hbox{$\smash{\scriptscriptstyle\frown}$}}{R} }_i}{{\overset{\lower0.5em\hbox{$\smash{\scriptscriptstyle\frown}$}}{U} }_k}||_{{T_S}}^2} \leq \frac{{\kappa \left( {{T_S}} \right)}}
{{{p_{\min }}}}||\overset{\lower0.5em\hbox{$\smash{\scriptscriptstyle\frown}$}}{R} {{\overset{\lower0.5em\hbox{$\smash{\scriptscriptstyle\frown}$}}{U} }_k}||_{{T_S}}^2 \hfill \\
  \qquad\qquad\qquad\quad\,\,\, \mathop  = \limits^{\eqref{Fuzhuite}} \frac{{\kappa \left( {{T_S}} \right)}}
{{{\eta ^2}{p_{\min }}}}||{{\tilde U}_{k + 1}} - {U_k}||_{{T_S}}^2 \hfill \\
\end{gathered}
\end{flalign}
For the third term in the last equality of \eqref{ConditrTUbar}, we have
\begin{flalign} \label{HuRiuu}
\begin{gathered}
\quad {\langle {{U^*} \!-\! {U_k},\overset{\lower0.5em\hbox{$\smash{\scriptscriptstyle\frown}$}}{R} {{\overset{\lower0.5em\hbox{$\smash{\scriptscriptstyle\frown}$}}{U} }_k}} \rangle _{{T_S}}} \hfill \\
  = {\langle {\overset{\lower0.5em\hbox{$\smash{\scriptscriptstyle\frown}$}}{R} {{\overset{\lower0.5em\hbox{$\smash{\scriptscriptstyle\frown}$}}{U} }_k} \!-\! R{U^ * },{U^*} \!- \! {{\overset{\lower0.5em\hbox{$\smash{\scriptscriptstyle\frown}$}}{U} }_k}} \rangle _{{T_S}}}  \!+\! {\langle {\overset{\lower0.5em\hbox{$\smash{\scriptscriptstyle\frown}$}}{R} {{\overset{\lower0.5em\hbox{$\smash{\scriptscriptstyle\frown}$}}{U} }_k},\!\!\! \sum\limits_{d \in J\left( k \right)} {\left( {{U_d} \!-\! {U_{d + 1}}} \right)} } \rangle _{{T_S}}} \hfill \\
  =  - {\langle {{{\overset{\lower0.5em\hbox{$\smash{\scriptscriptstyle\frown}$}}{U} }_k} - {U^*},\overset{\lower0.5em\hbox{$\smash{\scriptscriptstyle\frown}$}}{R} {{\overset{\lower0.5em\hbox{$\smash{\scriptscriptstyle\frown}$}}{U} }_k} - \overset{\lower0.5em\hbox{$\smash{\scriptscriptstyle\frown}$}}{R} {U^*}} \rangle _{{T_S}}} \hfill \\
  \quad + \frac{1}{\eta }\sum\limits_{d \in J\left( k \right)} {{{\langle {{U_d} - {U_{d + 1}},{U_k} - {{\tilde U}_{k + 1}}} \rangle }_{{T_S}}}}  \hfill \\
   \leq  - \frac{1}{2}||\overset{\lower0.5em\hbox{$\smash{\scriptscriptstyle\frown}$}}{R} {{\overset{\lower0.5em\hbox{$\smash{\scriptscriptstyle\frown}$}}{U} }_k}||_{{T_S}}^2 \hfill \\
  \quad + \frac{1}{{2\eta }}\sum\limits_{d \in J\left( k \right)} {(\frac{1}
{c}||{U_k} \!-\! {{\tilde U}_{k + 1}}||_{{T_S}}^2 + c||{U_d} - {U_{d + 1}}||_{{T_S}}^2)}  \hfill \\
=  - \frac{1}{{2{\eta ^2}}}||{{\tilde U}_{k + 1}} - {U_k}||_{{T_S}}^2 + \frac{{|J\left( k \right)|}}
{{2c\eta }}||{{\tilde U}_{k + 1}} - {U_k}||_{{T_S}}^2 \hfill \\
  \quad + \frac{c}{{2\eta }}\sum\limits_{d \in J\left( k \right)} {||{U_d} - {U_{d + 1}}||_{{T_S}}^2}  \hfill \\
\end{gathered}
\end{flalign}
where the first equality comes from Lemma \ref{OIYRE} and also involves that $R{U^ * } = 0$, the second one follows from \eqref{Fuzhuite}. Meanwhile, the inequality is derived from \cite[Lemma 3]{Tianyu2018Decentralized} and the basic inequality like \eqref{fertfw} for some $c>0$. Substituting \eqref{HuRiUU} and \eqref{HuRiuu} into \eqref{ConditrTUbar}, it arrives at
\begin{flalign*}
\begin{gathered}
  \quad \mathbb{E}\left( {||{U_{k + 1}} - {U^*}||_{{T_S}}^2|{\mathcal{Z}_k}} \right) \hfill \\
   \leq ||{U_k} - {U^*}||_{{T_S}}^2 + \frac{{{\eta ^2}}}{{{m^2}}}\frac{{\kappa \left( {{T_S}} \right)}}{{{\eta ^2}{p_{\min }}}}||{{\tilde U}_{k + 1}} - {U_k}||_{{T_S}}^2 \hfill \\
   \quad - \frac{1}{m}||{{\tilde U}_{k + 1}} - {U_k}||_{{T_S}}^2 + \frac{{|J\left( k \right)|}}{{cm}}||{{\tilde U}_{k + 1}} - {U_k}||_{{T_S}}^2 \hfill \\
  \quad + \frac{c}{m}\sum\limits_{d \in J\left( k \right)} {||{U_d} - {U_{d + 1}}||_{{T_S}}^2}  \hfill \\
  = ||{U_k} - {U^*}||_{{T_S}}^2 + \frac{c}{m}\sum\limits_{d \in J\left( k \right)} {||{U_d} - {U_{d + 1}}||_{{T_S}}^2}  \hfill \\
  \quad - \frac{1}{m}\left( {\frac{1}{\eta } - \frac{{|J\left( k \right)|}}{c} - \frac{{\kappa ({T_S})}}{{m{p_{\min }}}}} \right)||{{\tilde U}_{k + 1}} - {U_k}||_{{T_S}}^2 \hfill \\
\end{gathered}
\end{flalign*}
Then, this completes the proof.
\end{proof}

\subsection{Proof of Theorem \ref{Fundainequa}} \label{ProfFundainequa}
\begin{proof}
Let $c = m\sqrt {{p_{\min }}/\kappa \left( {{T_S}} \right)}$ in \eqref{rtoq}. Then, compute the conditional expectation of $\left\| {{{\mathbf{U}}_{k + 1}} - {{\mathbf{U}}^ * }} \right\|_\Phi ^2$ to obtain that
\begin{flalign} \label{HeitiUStarU}
\begin{gathered}
  \quad \mathbb{E}\left( {||{{\mathbf{U}}_{k + 1}} - {{\mathbf{U}}^ * }||_\Phi ^2|{\mathcal{Z}_k}} \right) \hfill \\
   = \mathbb{E}\left( {||{U_{k + 1}} - {U^ * }||_{{T_S}}^2|{\mathcal{Z}_k}} \right) \hfill \\
  \quad + \sqrt {\frac{{{p_{\min }}}}{{\kappa \left( {{T_S}} \right)}}} \sum\limits_{d = k + 1 - \varepsilon }^{k - 1} {(d - (k - \tau ))||{U_d} - {U_{d + 1}}||_{{T_S}}^2}  \hfill \\
  \quad + \varepsilon \sqrt {\frac{{{p_{\min }}}}{{\kappa \left( {{T_S}} \right)}}} \mathbb{E}(||{U_k} - {U_{k + 1}}||_{{T_S}}^2|{\mathcal{Z}_k}) \hfill \\
  \mathop  \leq \limits^{\eqref{rtoq}} \left\| {{U_k} - {U^*}} \right\|_{{T_S}}^2 \hfill \\
  \quad + \sqrt {\frac{{{p_{\min }}}}{{\kappa \left( {{T_S}} \right)}}} \sum\limits_{d = k - \varepsilon }^{k - 1} {(d \!-\! (k \!-\! \varepsilon ) \!+\! 1)||{U_d} \!-\! {U_{d + 1}}||_{{T_S}}^2}  \hfill \\
  \quad - \frac{1}{m}\left( {\frac{1}{\eta } \!-\! \sqrt {\frac{{\kappa \left( {{T_S}} \right)}}{{{p_{\min }}}}} \frac{{|J\left( k \right)| \!+\! \varepsilon }}{m}} \right)||{{\tilde U}_{k + 1}} \!-\! {U_k}||_{{T_S}}^2 \hfill \\
  \quad + \varepsilon \sqrt {\frac{{{p_{\min }}}}{{\kappa \left( {{T_S}} \right)}}} \mathbb{E}(||{U_k} - {U_{k + 1}}||_{{T_S}}^2|{\mathcal{Z}_k}) \hfill \\
  \mathop  = \limits^{\eqref{QiWangHeitiU}} \left\| {{{\mathbf{U}}_k}\! -\! {{\mathbf{U}}^ * }} \right\|_\Phi ^2 \!+\! \varepsilon \sqrt {\frac{{{p_{\min }}}}
{{\kappa \left( {{T_S}} \right)}}} \mathbb{E}(||{U_k} \!-\! {U_{k + 1}}||_{{T_S}}^2|{\mathcal{Z}_k}) \hfill \\
  \quad - \frac{1}{m}\left( {\frac{1}{\eta } \!-\! \sqrt {\frac{{\kappa \left( {{T_S}} \right)}}{{{p_{\min }}}}} \frac{{|J\left( k \right)| \!+\! \varepsilon }}{m}} \right)||{{\tilde U}_{k + 1}} \!-\! {U_k}||_{{T_S}}^2 \hfill \\
\end{gathered}
\end{flalign}
On the one hand, we derive the conditional expectation of $||{U_k} - {U_{k + 1}}||_{{T_S}}^2$ as follows:
\begin{flalign*}
\begin{gathered}
  \quad \mathbb{E}\left( {||{U_k} - {U_{k + 1}}||_{{T_S}}^2|{\mathcal{Z}_k}} \right) \hfill \\
  \mathop  = \limits^{\eqref{Fuzhuite}} \mathbb{E}(||{\eta_{{i_k}}}{{\overset{\lower0.5em\hbox{$\smash{\scriptscriptstyle\frown}$}}{R} }_{{i_k}}}{{\overset{\lower0.5em\hbox{$\smash{\scriptscriptstyle\frown}$}}{U} }_k}||_{{T_S}}^2|{\mathcal{Z}_k}) = \sum\limits_{i = 1}^m {{p_i}\frac{{{\eta ^2}}}
{{{m^2}p_i^2}}||{{\overset{\lower0.5em\hbox{$\smash{\scriptscriptstyle\frown}$}}{R} }_i}{{\overset{\lower0.5em\hbox{$\smash{\scriptscriptstyle\frown}$}}{U} }_k}||_{{T_S}}^2}  \hfill \\
  \mathop  \leq \limits^{\eqref{HuRiUU}} \frac{1}{{{m^2}}}\frac{{\kappa \left( {{T_S}} \right)}}{{{p_{\min }}}}||{{\tilde U}_{k + 1}} - {U_k}||_{{T_S}}^2 \hfill \\
\end{gathered}
\end{flalign*}
On the other hand, Lemma \ref{OIYRE} gives that $\left| {J\left( k \right)} \right| \leq \varepsilon $. Hence, plugging the aforementioned result into \eqref{HeitiUStarU} yields the desired result \eqref{HeiUkZ}. Therefore, the convergence of ${\{ {\| {{{\mathbf{U}}_{k + 1}} - {{\mathbf{U}}^ * }} \|_\Phi ^2} \}_{k \in \mathbb{N}}}$ follows from \cite[Lemma 13]{Zhimin2016ARock}.
\end{proof}
\end{appendices}

\bibliographystyle{ieeetr}


\bibliography{paper}

\begin{thebibliography}{10}

\bibitem{Ruud2020Solving}
R.~Egging-Bratseth, T.~Baltensperger, and A.~Tomasgard, ``{Solving
  oligopolistic equilibrium problems with convex optimization},'' {\em European
  Journal of Operational Research}, vol.~284, pp.~44--52, 2020.

\bibitem{Mao2017Price}
M.~Ye and G.~Hu, ``{Game design and analysis for price-based demand response:
  An aggregate game approach},'' {\em IEEE Transactions on Cybernetics},
  vol.~47, no.~3, pp.~720--730, 2017.

\bibitem{Liang2019Generalized}
Y.~Liang, W.~Wei, and C.~Wang, ``{A generalized Nash equilibrium approach for
  autonomous energy management of residential energy hubs},'' {\em IEEE
  Transactions on Industrial Informatics}, vol.~15, no.~11, pp.~5892--5905,
  2019.

\bibitem{Ankur2012variational}
A.~A. Kulkarni and U.~V. Shanbhag, ``{On the variational equilibrium as a
  refinement of the generalized Nash equilibrium},'' {\em Automatica}, vol.~48,
  no.~1, pp.~45--55, 2012.

\bibitem{Patrick1991Generalized}
P.~T. Harker, ``{Generalized Nash games and quasi-variational inequalities},''
  {\em European Journal of Operational Research}, vol.~54, no.~1, pp.~81--94,
  1991.

\bibitem{Francisco2007generalized}
F.~Facchinei, A.~Fischer, and V.~Piccialli, ``{On generalized Nash games and
  variational inequalities},'' {\em Operations Research Letters}, vol.~35,
  no.~2, pp.~159--164, 2007.

\bibitem{Kaemmler2015GENERALIZED}
M.~Hinterm{\"{u}}ller, T.~Surowiec, and A.~K{\"{a}}mmler, ``{Generalized Nash
  equilibrium problems in Banach spaces: Theory, nikaido-isoda-based
  path-following methods, and applications},'' {\em SIAM Journal on
  Optimization}, vol.~25, no.~3, pp.~1826--1856, 2015.

\bibitem{Jason2009Joint}
G.~A. Jason R.~Marden and J.~S. Shamma, ``{Joint strategy fictitious play with
  inertia for potential games},'' {\em IEEE Transactions on Automatic Control},
  vol.~54, no.~2, pp.~208--220, 2009.

\bibitem{Mattia2020Proximal}
M.~Bianchi, G.~Belgioioso, and S.~Grammatico, ``{A fully-distributed
  proximal-point algorithm for Nash equilibrium seeking with linear convergence
  rate},'' in {\em Proceedings of the IEEE Conference on Decision and Control},
  pp.~2303--2308, 2020.

\bibitem{Belgioioso2021time}
G.~Belgioioso, A.~Nedi{\'{c}}, and S.~Grammatico, ``{Distributed generalized
  Nash equilibrium seeking in aggregative games on time-varying networks},''
  {\em IEEE Transactions on Automatic Control}, vol.~66, no.~5, pp.~2061--2075,
  2021.

\bibitem{Distributed2021Yanan}
Y.~Zhu, W.~Yu, G.~Wen, and G.~Chen, ``{Distributed Nash equilibrium seeking in
  an aggregative game on a directed Graph},'' {\em IEEE Transactions on
  Automatic Control}, vol.~66, no.~6, pp.~2746--2753, 2021.

\bibitem{Maojiao2022Differentially}
M.~Ye, G.~Hu, L.~Xie, and S.~Xu, ``{Differentially private distributed Nash
  equilibrium seeking for aggregative games},'' {\em IEEE Transactions on
  Automatic Control}, vol.~67, no.~5, pp.~2451--2458, 2022.

\bibitem{Yipeng2021Distributed}
Y.~Pang and G.~Hu, ``{Distributed Nash equilibrium seeking with limited cost
  function knowledge via a consensus-based gradient-free method},'' {\em IEEE
  Transactions on Automatic Control}, vol.~66, no.~4, pp.~1832--1839, 2021.

\bibitem{Fei2021Distributed}
F.~Liu, Q.~Wang, Y.~Hua, X.~Dong, and Z.~Ren, ``{Distributed Nash equilibrium
  seeking for non-cooperative convex games with local constraints},'' in {\em
  40th Chinese Control Conference (CCC)}, pp.~7480--7485, 2021.

\bibitem{Salehisadaghiani2019Seeking}
F.~Salehisadaghiani, W.~Shi, and L.~Pavel, ``{Distributed Nash equilibrium
  seeking under partial-decision information via the alternating direction
  method of multipliers},'' {\em Automatica}, vol.~103, pp.~27--35, 2019.

\bibitem{Kaihong2021Nonsmooth}
K.~Lu and Q.~Zhu, ``{Nonsmooth continuous-time distributed algorithms for
  seeking generalized Nash equilibria of noncooperative games via digraphs},''
  {\em IEEE Transactions on Cybernetics}, 2021.

\bibitem{Kaihong2019Distributed}
K.~Lu, G.~Jing, and L.~Wang, ``{Distributed algorithms for searching
  generalized Nash equilibrium of noncooperative games},'' {\em IEEE
  Transactions on Cybernetics}, vol.~49, no.~6, pp.~2362--2371, 2019.

\bibitem{Liang2017Coupled}
S.~Liang, P.~Yi, and Y.~Hong, ``{Distributed Nash equilibrium seeking for
  aggregative games with coupled constraints},'' {\em Automatica}, vol.~85,
  pp.~179--185, 2017.

\bibitem{Yi2019splitting}
P.~Yi and L.~Pavel, ``{An operator splitting approach for distributed
  generalized Nash equilibria computation},'' {\em Automatica}, vol.~102,
  pp.~111--121, 2019.

\bibitem{Franci2021Distributed}
B.~Franci and S.~Grammatico, ``{A distributed forward-backward algorithm for
  stochastic generalized Nash equilibrium seeking},'' {\em IEEE Transactions on
  Automatic Control}, vol.~66, no.~11, pp.~5467--5473, 2021.

\bibitem{Pavel2020Distributed}
L.~Pavel, ``{Distributed GNE seeking under partial-decision information over
  networks via a doubly-augmented operator splitting approach},'' {\em IEEE
  Transactions on Automatic Control}, vol.~65, no.~4, pp.~1584--1597, 2020.

\bibitem{Zhaojian2021Energy}
Z.~Wang, F.~Liu, Z.~Ma, Y.~Chen, M.~Jia, W.~Wei, and Q.~Wu, ``{Distributed
  generalized Nash equilibrium seeking for energy sharing games in
  prosumers},'' {\em IEEE Transactions on Power Systems}, vol.~36, no.~5,
  pp.~3973--3986, 2021.

\bibitem{Tian2020Achieving}
Y.~Tian, Y.~Sun, and G.~Scutari, ``{Achieving linear convergence in distributed
  asynchronous multiagent optimization},'' {\em IEEE Transactions on Automatic
  Control}, vol.~65, no.~12, pp.~5264--5279, 2020.

\bibitem{Tianyu2018Decentralized}
T.~Wu, K.~Yuan, Q.~Ling, W.~Yin, and A.~H. Sayed, ``Decentralized consensus
  optimization with asynchrony and delays,'' {\em {IEEE Transactions on Signal
  and Information Processing over Networks}}, vol.~4, no.~2, pp.~293--307,
  2018.

\bibitem{Lei2017Randomized}
J.~Lei, U.~Shanbhag, J.-S. Pang, and S.~Sen, ``{On synchronous, asynchronous,
  and randomized best-response schemes for computing equilibria in stochastic
  Nash games},'' {\em Mathematics of Operations Research}, 2017.

\bibitem{Peng2020Asynchronous}
P.~Yi and L.~Pavel, ``{Asynchronous distributed algorithms for seeking
  generalized Nash equilibria under full and partial-decision information},''
  {\em IEEE Transactions on Cybernetics}, vol.~50, no.~6, pp.~2514--2526, 2020.

\bibitem{Cenedese2019An}
C.~Cenedese, G.~Belgioioso, S.~Grammatico, and M.~Cao, ``{An asynchronous,
  forward-backward, distributed generalized Nash equilibrium seeking
  algorithm},'' in {\em 2019 18th European Control Conference, ECC 2019}, 2019.

\bibitem{Francisco2010Generalized}
F.~Facchinei and C.~Kanzow, ``{Generalized Nash equilibrium problems},'' {\em
  Annals of Operations Research}, vol.~175, pp.~177--211, 2010.

\bibitem{Belgioioso2017Decentralized}
G.~Belgioioso and S.~Grammatico, ``{Semi-decentralized Nash equilibrium seeking
  in aggregative games with separable coupling constraints and
  non-differentiable cost functions},'' {\em IEEE Control Systems Letters},
  vol.~1, no.~2, pp.~400--405, 2017.

\bibitem{Zheng2023Stochastic}
L.~Zheng, H.~Li, L.~Ran, L.~Gao, and D.~Xia, ``{Distributed primal¨Cdual
  algorithms for stochastic generalized Nash equilibrium seeking under full and
  partial-decision information},'' {\em IEEE Transactions on Control of Network
  Systems}, vol.~10, no.~2, pp.~718--730, 2023.

\bibitem{Kannan2012Iterative}
A.~Kannan and U.~V. Shanbhag, ``{Distributed computation of equilibria in
  monotone Nash games via iterative regularization techniques},'' {\em SIAM
  Journal on Optimization}, vol.~22, no.~3, pp.~1177--1205, 2012.

\bibitem{Scutari2010Convex}
G.~Scutari, D.~P. Palomar, F.~Facchinei, and J.-s. Pang, ``Convex optimization,
  game theory, and variational inequality theory,'' {\em IEEE Signal Processing
  Magazine}, vol.~27, no.~3, pp.~35--49, 2010.

\bibitem{Facchinei2009Nash}
F.~Facchinei and J.-S. Pang, ``{Nash equilibria: The variational approach},''
  {\em Convex Optimization in Signal Processing and Communications},
  pp.~443--493, 2009.

\bibitem{Bauschke2011Convex}
H.~H. Bauschke and P.~L. Combettes, ``{Convex analysis and monotone operator
  theory in Hilbert spaces},'' in {\em Cham, Switzerland: Springer}, 2011.

\bibitem{Carnevale2022Carnevale}
G.~Carnevale, F.~Fabiani, F.~Fele, K.~Margellos, and G.~Notarstefano,
  ``{Tracking-based distributed equilibrium seeking for aggregative games},''
  {\em arXiv:2210.14547}, 2022.

\bibitem{Mansoori2020Fast}
M.~Fatemeh and W.~Ermin, ``{A fast distributed asynchronous newton-based
  optimization algorithm},'' {\em IEEE Transactions on Automatic Control},
  vol.~65, no.~7, pp.~2769--2784, 2020.

\bibitem{Zhimin2016ARock}
Z.~Peng, Y.~Xu, M.~Yan, and W.~Yin, ``{ARock: An algorithmic framework for
  asynchronous parallel coordinate updates},'' {\em SIAM Journal on Scientific
  Computing}, vol.~38, no.~5, pp.~2851--2879, 2016.

\bibitem{Latafat2017Asymmetric}
P.~Latafat and P.~Patrinos, ``{Asymmetric forward-backward-adjoint splitting
  for solving monotone inclusions involving three operators},'' {\em
  Computational Optimization and Applications}, vol.~68, pp.~57--93, 2017.

\bibitem{Yan2018A}
M.~Yan, ``{A new primal-dual algorithm for minimizing the sum of three
  functions with a linear operator},'' {\em Journal of Entific Computing},
  vol.~76, no.~3, pp.~1698--1717, 2018.

\bibitem{Opial1967Weak}
Z.~Opial, ``{Weak convergence of the sequence of successive approximations for
  nonexpansive mappings},'' {\em Bulletin of the American Mathematical
  Society}, vol.~73, no.~4, pp.~591--597, 1967.

\bibitem{Guannan2020Accelerated}
G.~Qu and N.~Li, ``{Accelerated distributed Nesterov gradient descent},'' {\em
  IEEE Transactions on Automatic Control}, vol.~65, no.~6, pp.~2566--2581,
  2020.

\end{thebibliography}

\end{document}